\documentclass[prodmode]{acmsmall}
\acmVolume{X}
\acmNumber{X}
\acmArticle{X}
\acmYear{2014}
\acmMonth{01}
\usepackage[numbers]{natbib}
\usepackage{amssymb}

\RequirePackage{ifthen,calc}
\newsavebox{\ffbox}\newlength{\ffboxlen}
\newcommand{\todo}[1]{%
  {\sbox{\ffbox}{\textbf{TODO:}\ \textit{{#1}}\ \textbf{:ODOT}}
    \settowidth{\ffboxlen}{\usebox{\ffbox}}
		\addtolength{\ffboxlen}{-5mm}
    \ifthenelse{\ffboxlen>\linewidth}{%
      \noindent\marginpar{$>>>>$}\textbf{TODO:}\ \textit{{#1}}\ \textbf{:ODOT}\marginpar{$<<<<$}}{%
      \noindent\marginpar{$>><<$}\textbf{TODO:}\ \textit{{#1}}\ \textbf{:ODOT}}}}

\usepackage{array,xspace,multirow,hhline,tikz,tabularx,booktabs,fixltx2e}

\usetikzlibrary{trees}
\usetikzlibrary{positioning,chains,fit,shapes,calc}
\usepackage{pgflibraryshapes}

\usepackage{subfig}
\usepackage{colortbl,color,array,hhline}
\usepackage{tikz}
\usetikzlibrary{arrows}
\usetikzlibrary{decorations.pathreplacing}
\usepackage{lscape}

\DeclareSymbolFont{largesymbolsA}{U}{txexa}{m}{n}
\DeclareMathSymbol{\varprod}{\mathop}{largesymbolsA}{"10}

\usepackage{amsmath,amssymb,amsfonts}

\usetikzlibrary{arrows}
\usetikzlibrary{decorations.pathreplacing}
\usetikzlibrary{patterns}

\usepackage{paralist}
\usepackage[boxed]{algorithm}
\usepackage{algorithmic}

	\newlength{\wordlength}
	\newcommand{\wordbox}[3][c]{\settowidth{\wordlength}{#3}\makebox[\wordlength][#1]{#2}}

\newcommand{\eg}{e.g.,\xspace}

\newcommand{\pref}{\succ \xspace}
\newcommand{\preffn}[2]{\; \pref_{#1}\!\!\!({#2}) \xspace}

\newcommand{\prefinvfn}[2]{\; \pref^{-1}_{#1}\!\!({#2})}

\newcommand{\prefcloned}{\succ^{\textsc{cloned}} \xspace}

\newcommand{\hclone}{H^{\textsc{cloned}} \xspace}

\newcommand{\prefbest}{\pref^{\textsc{best}}}

\newcommand{\bestalloc}{\textsc{Best-Alloc}}

\newcommand{\prefclonedbest}{\succ^{\textsc{cloned-best}} \xspace}

\newcommand{\Indiff}[1][]{
	\ifthenelse{\equal{#1}{}}{\mathrel I}{\mathop{I_{#1}}}
}
\newcommand{\prefset}[1][]{\ifthenelse{\equal{#1}{}}{\mathcal{R}}{\mathcal{R}_{#1}}}

\usepackage{enumitem}
\setenumerate[1]{label=\rm(\it{\roman{*}}\rm),ref=({\it\roman{*}}),leftmargin=*}

\newtheorem{claim}{Claim}

\usepackage[mathcal]{euscript}

\newcommand{\midd}{\mathrel{:}}

\newcommand{\set}[1]{\{#1\}}

		\newcommand{\sPref}[1][]{                  
			\ifthenelse{\equal{#1}{}}{\mathrel P}{\mathop{P_{#1}}}
		}

				\newcommand{\ps}{\mbox{\sc PS}}

		\newcommand{\BestEURes}{\mbox{\sc BestEUresponse}}
		\newcommand{\BestEUResAlgo}{\mbox{\sc BestEUresponseAlgo}}
		
		\tikzset{
		  treenode/.style = {align=center, inner sep=0pt, text centered,
		    font=\sffamily},
		  arn_n/.style = {treenode, circle, white, font=\sffamily\bfseries, draw=black,
		    fill=black, text width=1.5em},
		  arn_r/.style = {treenode, circle, red, draw=red,
		    text width=1.5em, very thick},
		  arn_x/.style = {treenode, rectangle, draw=black,
		    minimum width=0.5em, minimum height=0.5em}
		}

\author{HARIS AZIZ and SERGE GASPERS and NICK MATTEI and TOBY WALSH\affil{NICTA and University of New South Wales, Australia} and NINA NARODYTSKA\affil{University of Toronto, Canada} 
}

\title{Strategic aspects of the probabilistic serial rule\\ for the allocation of goods}

\begin{abstract}
	The probabilistic serial (PS) rule is one of the most prominent randomized rules for the assignment problem. It is well-known for its superior fairness and welfare properties. However, PS is not immune to manipulative behaviour by the agents. We examine computational and non-computational aspects of strategising under the PS rule. Firstly, we study the computational complexity of an agent manipulating the PS rule. We present polynomial-time algorithms for optimal manipulation. 
	Secondly, we show that expected utility best responses can cycle.
	Thirdly, we examine the existence and computation of Nash equilibrium profiles under the PS rule. We show that a pure Nash equilibrium is guaranteed to exist under the PS rule. For two agents, we identify two different types of preference profiles that are not only in Nash equilibrium but can also be computed in linear time. Finally, we conduct experiments to check the frequency of manipulability of the PS rule under different combinations of the number of agents, objects, and utility functions.
\end{abstract}




\begin{document}
	\begin{bottomstuff}
	Emails: \texttt{haris.aziz@nicta.com.au},
	\texttt{sergeg@cse.unsw.edu.au}, \texttt{Nicholas.Mattei@nicta.com.au},
	\texttt{ninan@cs.toronto.edu}, \texttt{toby.walsh@nicta.com.au} 
	\end{bottomstuff}


	\sloppy

	\maketitle

	\section{Introduction}

	The \emph{assignment problem} is one of the most fundamental and important problems in economics and computer science~\citep[see \eg][]{BoMo01a,Gard73b,HyZe79a,ABS13a, SeSa13a}. Agents express preferences over objects and, based on these preferences, the objects are allocated to the agents.
	A randomized or fractional assignment rule takes the preferences of the agents into account in order to allocate each agent a fraction of the object. If the objects are indivisible, the fraction can also be interpreted as the probability of receiving the object. Randomization is widespread in resource allocation since it is one of the most natural ways to ensure procedural fairness~\citep{BCKM12a}. Randomized assignments have been used to assign public land, radio spectra to broadcasting companies, and US permanent visas to applicants~\citep[Footnote~1 in ][]{BCKM12a}.

	Typical criteria for randomized assignment being desirable are fairness and welfare.
	The \emph{probabilistic serial (PS)} rule is an ordinal randomized/fractional assignment rule that fares better on both counts than any other random assignment rule~\citep{BoHe12a,BoMo01a,BCKM12a, KaSe06a,Koji09a, Yilm10a,SaSe13b}.
	In particular, it satisfies strong envy-freeness and efficiency with respect to both \emph{stochastic dominance (SD)} and \emph{downward lexicographic (DL)} relations~\citep{BoMo01a,ScVa12a,Koji09a}.
	  SD is one of the most fundamental relations between fractional allocations because one allocation is SD-preferred over another iff for any utility representation consistent with the ordinal preferences, the former yields at least as much expected utility as the latter. DL is a refinement of SD and based on lexicographic comparisons between fractional allocations.
	Generalizations of the PS rule have been recommended in many settings~\citep[see \eg][]{BCKM12a}.
	The PS rule also satisfies some desirable incentive properties. If the number of objects is not more than the number of agents, then PS is weak strategyproof with respect to stochastic dominance~\citep{BoMo01a}. 
	However, PS is not immune from manipulation.\footnote{Another well-established rule \textit{random serial dictator (RSD)} is strategyproof but it is not envy-free and not as efficient as PS~\citep{BoMo01a}. Moreover, in contrast to PS, the fractional allocations under RSD  are \#P-complete to compute~\citep{ABB13b}.}

	PS works as follows. Each agent expresses linear orders over the set of houses (we use the term house throughout the paper though we stress any object could be allocated with these mechanisms). Each house
	is considered to have a divisible probability weight of one, and agents simultaneously and with the same speed consume the probability weight of their most preferred house. Once a house has been consumed, the agent proceeds to eat the next most preferred house that has not been completely consumed. The procedure terminates after all the houses have been consumed. The random allocation of an agent by PS is the amount of each object he has eaten.\footnote{Although PS was originally defined for the setting where the number of houses is equal to the number of agents, it can be used without any modification for fewer or more houses than agents~\citep[see \eg][]{BoMo01a,Koji09a}.}

	We examine the following natural questions for the first time: \emph{what is the computational complexity of an agent computing a different preference to report so as to get a better PS outcome? How often is a preference profile manipulable under the PS rule?}. 
\footnote{This problem of computing the optimal manipulation has already been studied in great depth for voting rules~\citep[see \eg][]{FaPr10a,FHH10a}.}
	The complexity of manipulation of the PS rule has bearing on another issue that has recently been studied---preference profiles that are in Nash equilibrium.  \citet{EkKe12a} showed that when agents are not truthful, the outcome of PS may not satisfy desirable properties related to efficiency and envy-freeness. 
	Because the PS rule is manipulable it is important to understand how hard, computationally, it is for an agent
	to compute a beneficial misreporting as this may make it difficult in practice to exploit the mechanism.
	It is also interesting to identify  preference profiles for which no agent has an incentive to unilaterally deviate to gain utility with respect to his actual preferences.
	Hence, we consider the following problem: \emph{for a preference profile, does a (pure) Nash equilibrium exist or not and if it exists how efficiently can it be computed?}
	%

	In order to compare random allocations, an agent needs to consider relations between random allocation. We consider three well-known relations between lotteries~\citep[see \eg][]{BoMo01a,ScVa12a,SaSe13b,Cho12a}:
	$(i)$ \textit{expected utility (EU)}, $(ii)$ \textit{stochastic dominance (SD)},  and $(iii)$ \textit{downward lexicographic (DL)}.
	For EU, an agent seeks a different allocation that yields more expected utility. For SD, an agent seeks a different allocation that yields more expected utility for all cardinal utilities consistent with the ordinal preferences. For DL, an agent seeks an allocation that gives a higher probability to the most preferred alternative that has different probabilities in the two allocations.
	Throughout the paper, we assume that agents express \emph{strict} preferences, i.e., they are not indifferent between any two houses. 
	

	\paragraph{Contributions}

	We initiate the study of computing best responses and checking for Nash equilibrium for the PS mechanism --- one of the most established randomized rules for the assignment problem.
	We present a polynomial-time algorithm to compute the DL best response for multiple agents and houses. The algorithm
	works by carefully simulating the PS rule for a sequence of partial preference lists. 
	For the case of two agents\footnote{The two-agent case is also of special importance since various disputes arise between two parties.}, we present a polynomial-time algorithm to compute an EU best response for any utilities consistent with the ordinal preferences.
	The result for the EU best response relies on an interesting connection between the PS rule and the sequential allocation rule for discrete objects. We leave open the problem of computing the expected utility response for arbitrary number of agents. The fact that a similar problem has also remained open for sequential allocation~\citep{BoLa11a} gives some indication of the challenge of the problem.

	We then examine situations in which all agents are strategic.
	We first show that expected utility best responses can cycle. Nash dynamics in matching theory has been active area of research especially for the stable matching problem~\citep[see \eg][]{AGM+11a}.
	We then prove that a (pure) Nash equilibrium exists for any number of agents and houses.
	To the best of our knowledge, this is the first proof of the existence of a Nash equilibrium for the PS rule. 
	For the case of two agents we present two different linear-time algorithms to compute a preference profile that 
	is in Nash equilibrium with respect to the original preferences. One type of equilibrium profile 
results in the \emph{same} assignment as the one by original profile.

	Finally, we perform an experimental study of the frequency of manipulability of the PS mechanism.
	We investigate, under a variety of utility functions and preference distributions, the likelihood that some agent in a profile
	has an incentive to misreport his preference. The experiments identify settings and utility models in which PS is less susceptible to manipulation.




	\section{Preliminaries}


	An assignment problem $(N, H, \pref)$ consists  of a set of agents $N=\{1,\ldots, n\}$, a set of houses $H=\{h_1, \ldots, h_m\}$ and a preference profile $\pref=(\pref_1,\ldots, \pref_n)$ in which $\pref_i$ denotes a complete, transitive and strict ordering on $H$ representing the preferences of agent $i$ over the houses in  $H$. Since each $\pref_i$ will be strict throughout the paper, we will also refer to it simply as $\succ_i$.

	A fractional assignment is a $(n\times m)$ matrix $[p(i)(j)]$ such that for all $i\in N$, and $h_j\in H$, $0\leq p(i)(j)\leq 1$;  and for all $j\in \{1,\ldots, n\}$, $\sum_{i\in N}p(i)(j)= 1$ 
	The value $p(i)(j)$ is the fraction of house $h_j$ that agent $i$ gets. Each row $p(i)=(p(i)(1),\ldots, p(i)(m))$ represents the allocation of agent $i$.
A fractional assignment can also be interpreted as a random assignment where $p(i)(j)$ is the probability of agent $i$ getting house $h_j$. We will also denote $p(i)(j)$ by $p(i)(h_j)$.

	\paragraph{Relations between random allocations}
		A standard method to compare lotteries is to use the \emph{SD (stochastic dominance)} relation. 
		%
		 Given two random assignments $p$ and $q$, $p(i) \succ_i^{SD} q(i)$ i.e.,  a player $i$ \emph{SD~prefers} allocation $p(i)$ to $q(i)$ if
		$\sum_{h_j\in \set{h_k\midd h_k\pref_i h}}p(i)(h_j) \ge \sum_{h_j\in \set{h_k\midd h_k\pref_i h}}q(i)(h_j)$  for all  $h\in H$ and 
		$\sum_{h_j\in \set{h_k\midd h_k\pref_i h}}p(i)(h_j) > \sum_{h_j\in \set{h_k\midd h_k\pref_i h}}q(i)(h_j) \text{ for some } h\in H.$


		Given two random assignments $p$ and $q$, $p(i) \pref_i^{DL} q(i)$ i.e.,  a player $i$ \emph{DL~prefers} allocation $p(i)$ to $q(i)$ if $p(i)\neq q(i)$ and for the most preferred house $h$ such that $p(i)(h)\neq q(i)(h)$, we have that $p(i)(h)>q(i)(h)$.



			When agents are considered to have cardinal utilities for the objects, we denote by $u_i(h)$ the utility that agent $i$ gets from house $h$. We will assume that total utility of an agent equals the sum of the utilities that he gets from each of the houses. Given two random assignments $p$ and $q$, $p(i) \pref_i^{EU} q(i)$ i.e.,  a player $i$ \emph{EU (expected utility)~prefers} allocation $p(i)$ to $q(i)$ iff
		$\sum_{h\in H}u_i(h)p(i)(h)> \sum_{h\in H}u_i(h)q(i)(h).$

	Since for all $i\in N$, agent $i$ compares assignment $p$ with assignment $q$ only with respect to his allocations $p(i)$ and $q(i)$, we will sometimes abuse the notation and use $p\pref_i^{SD} q$ for $p(i)\pref_i^{SD} q(i)$.
	A \emph{random assignment rule} takes as input an assignment problem $(N,H,\pref)$ and returns a random assignment which specifies how much fraction or probability of each house is allocated to each agent.


	%

	\section{The Probabilistic Serial Rule and its Manipulation}

	Recall that the \emph{Probabilistic Serial (PS) rule} is a random assignment algorithm in which we consider each house as infinitely divisible.
	At each point in time, each agent is consuming his most preferred house that has not completely been consumed and each agent has the same unit speed. Hence all the houses are consumed at time $m/n$ and each agent receives a total of $m/n$ unit of houses.
	The probability of house $h_j$ being allocated to $i$ is the fraction of house $h_j$ that $i$ has eaten. The PS fractional assignment can be computed in time $O(mn)$.
	We refer the reader to \citep{BoMo01a} or \citep{Koji09a} for alternative definitions of PS. The following example adapted from \citep[Section 7, ][]{BoMo01a} shows how PS works.

	\begin{example}[PS rule]\label{example:PS}
		Consider an assignment problem with the following preference profile.
	\begin{align*}
		\succ_1:\quad& h_1,h_2,h_3 & \succ_2:\quad& h_2,h_1,h_3&	\succ_3:\quad& h_2,h_3,h_1
		\end{align*}
		Agents $2$ and $3$ start eating $h_2$ simultaneously whereas agent $1$ eats $h_1$. When $2$ and $3$ finish $h_2$, agent $3$ has only eaten half of $h_1$.  The timing of the eating can be seen below.
	\begin{center}
	             \begin{tikzpicture}[scale=0.2]
	                 \centering
	                 \draw[-] (0,0) -- (0,6);
	                 \draw[-] (0,0) -- (20,0);

	                 \draw[-] (20,6) -- (20,0);

	\draw[-] (0,2) -- (20,2);
	\draw[-] (0,4) -- (20,4);
	\draw[-] (20,0) -- (20,6);

	\draw[-] (10,0) -- (10,6);

	\draw[-] (0,6) -- (20,6);

	\draw[-] (15,0) -- (15,6);

	                                        \draw (0,-.6) node(c){\small $0$};
	                             \draw (20/2,-1) node(c){\small $\frac{1}{2}$};

	 \draw (20/2,-2.5) node(c){\small Time};

	                             \draw (20,-1) node(c){\small$1$};

	\draw (15,-1) node(c){\small$\frac{3}{4}$};

	    \draw(-3,6) node(z){\small Agent $1$};
	                 \draw(-3,4) node(z){\small Agent $2$};
	                 \draw(-3,2) node(z){\small Agent $3$};

	\draw(5,6.6) node(z){\small $h_1$};

	\draw(5,4.6) node(z){\small $h_2$};

	\draw(5,2.6) node(z){\small $h_2$};

	\draw(12.5,6.6) node(z){\small $h_1$};

	\draw(12.5,4.6) node(z){\small $h_1$};

	\draw(12.5,2.6) node(z){\small $h_3$};

	\draw(17.5,6.6) node(z){\small $h_3$};

	\draw(17.5,4.6) node(z){\small $h_3$};

	\draw(17.5,2.6) node(z){\small $h_3$};
	  \end{tikzpicture}
	\end{center}

		The final allocation computed by PS is
	$
	PS(\succ_1,\succ_2,\succ_3)=\begin{pmatrix}
		3/4&0&1/4\\
	  1/4&1/2& 1/4 \\
	  0&1/2 &  1/2

	 \end{pmatrix}.
	$

	\end{example}

	%
	%
	%
	%
	%
	%
	%
	%
	%

		Consider the assignment problem in Example~\ref{example:PS}. If agent $1$ misreports his preferences as follows: $\succ_1':\quad h_2,h_1,h_3,$ then $
	PS(\succ_1',\succ_2,\succ_3)=\begin{pmatrix}
		1/3&1/2&1/6\\
	  1/3 & 1/2 & 1/6 \\
	  1/3 & 0 & 2/3
		\end{pmatrix}.
	$
	\noindent
	Then, if $u_1(h_1)=7$, $u_1(h_2)=6$, and $u_1(h_3)=0$, then agent $1$ gets more expected utility when he reports $\succ_1'$. In the example, although truth-telling is a DL best response, it is not necessarily  an EU best response for agent $1$.

	Examples 1 and 2 of \citep{Koji09a} show that manipulating the PS mechanism can lead to an SD improvement when each agent can be allocated more than one house. In light of the fact that the PS rule can be manipulated, we examine the complexity of a single agent computing a manipulation, in other words, the best response for the PS rule.\footnote{Note that if an agent is risk-averse and does not have information about the other agent's preferences, then his maximin strategy is to be truthful. The reason is that if all all agents have the same preferences, then the optimal strategy is to be truthful.} We then study the existence and computation of Nash equilibria. For $\mathcal{E}\in \{SD, EU, DL\}$, we define the problem \textsc{$\mathcal{E}$BestResponse}: given $(N,H,\pref)$ and agent $i\in N$,
	 compute a preference $\pref_i'$ for agent $i$  such that there exists no preference $\pref_i''$ such that $PS(N,H,(\pref_i'',\pref_{-i})) \succ_i^{\mathcal{E}} PS(N,H,(\pref_i',\pref_{-i}))$.
	For a constant $m$, the problem \textsc{$\mathcal{E}$BestResponse} can can be solved by brute force by trying out each of the $m!$ preferences. Hence we won't assume that $m$ is a constant.

	We establish some more notation and terminology for the rest of the paper.
	We will often refer to the PS outcomes for partial lists of houses and preferences.
	We will denote by $PS(\pref_i^{L}, \pref_{-i})(i)$, the allocation that agent $i$ receives when
	his preferences are restricted to the list $L$ where $L$ is an ordered list of a subset of houses.  When an agent runs out of houses
	in his preference list, he does not eat any other houses.
	The \emph{length} of a list $L$ is denoted $|L|$, and we refer to the $k$th house in $L$ as $L(k)$.
	In the PS rule, the \emph{eating start time} of a house is the time point at which the house starts to be eaten by some agent. In Example~\ref{example:PS}, the eating start times of $h_1,h_2$ and $h_3$ are $0,0$ and $0.5$, respectively.

	\section{Lexicographic best response}
	\label{sec:dl}

	In this section, we present a polynomial-time algorithm for \textsc{DLBestResponse}. Lexicographic preferences are well-established in the assignment literature~\citep[see \eg][]{SaSe13b,ScVa12a,Cho12a}. 
Let $(N,H,\pref)$ be an assignment problem where $N=\{1,\dots,n\}$ and $H=\{h_1,\dots,h_m\}$. We will show how to compute a DL best response for agent $1\in N$. It has been shown that when $m\leq n$, then truth-telling is the DL best response but if $m>n$, then this need not be the case~\citep{SaSe13b,ScVa12a,Koji09a}.

	Recall that a preference $\succ_1'$ is a DL best response for agent 1 if the fractional allocation agent 1 receives by reporting $\succ_1'$ is DL preferred to any fractional allocation agent 1 receives by reporting another preference.
	That is, there is no preference $\succ_1''$ such that his share of a house $h$ when reporting $\succ_1''$ is strictly larger than when reporting $\succ_1'$ while the share of all houses he prefers to $h$ (according to his true preference $\succ_1$) is the same whether reporting $\succ_1'$ or $\succ_1''$.

	Our algorithm will iteratively construct a partial preference list for the $i$ most preferred houses of agent 1.
	Without loss of generality, denote
	 $\succ_1: h_1, h_2, \dots, h_m.$

	For any $i, 1\le i\le m$, denote $H_i = \{h_1, \dots, h_i\}$.
	A (partial) preference of agent 1 \emph{restricted} to $H_i$ is a preference over a subset of $H_i$.
	Note that a preference for $H_i$ need not list all the houses in $H_i$.
	For the preference of agent 1 restricted to $H_i$, the PS rule computes an allocation where the preference of agent 1 is replaced with this preference and the preferences of all other agents remain unchanged.
	Recall that agent 1 can only be allocated a non-zero fraction of a house if this house is in the preference list he submits.
	The notions of DL best response and DL preferred fractional assignments with respect to a subset of houses $H_i$ are defined accordingly for restricted preferences of agent 1.

	For a house $h\in H$, let $PS1(L,h)$ denote the fraction of house $h$ that the PS rule assigns to agent 1 when he reports the (partial) preference $L$.

	We start with a simple lemma showing that a DL best response for agent 1 for the whole set $H$ can be no better and no worse on $H_i$ than a DL best response for $H_i$.

	\begin{lemma}\label{lem:eq-ass}
	 Let $i\in\{1,\dots,m\}$. A DL best response for agent 1 on $H$ gives the same fractional assignment to the houses in $H_i$ as a DL best response for agent 1 on $H_i$.
	\end{lemma}
	\begin{proof}
	 We have that a preference for agent 1 on $H_i$ can be extended to a preference for all houses that gives the same fractional allocation to agent 1 for the houses in $H_i$.
	 Namely, the remaining houses $H\setminus H_i$ can be appended to the end of his preference list, giving the same allocation to the houses in $H_i$ as before.

	 On the other hand, consider a DL best response $\succ_1'$ for agent 1 on $H$, giving a fractional allocation $p$ to agent 1.
	 Restricting this preference to $H_i$ gives a fractional allocation $q$ for $H_i$.
	 If $q$ is DL preferred to $p_{|H_i}$, i.e., the fractional allocation $p$ restricted to $H_i$, then $q=p_{|H_i}$, otherwise we would have a contradiction to $\succ_1'$ being a DL best response as per the previous argument that we can extend any preference for $H_i$ to $H$ giving the same fractional allocation to agent 1 for the houses in $H_i$.
	\end{proof}

	\noindent
	Our algorithm will compute a list $L_i$ such that $L_i \subseteq H_i$.\footnote{When we treat a list as a set we refer to the set of all elements occurring in the list.}
	The list $L_i$ will be a DL best response for agent 1 with respect to $H_i$.
	Suppose the algorithm has computed $L_{i-1}$.
	Then, when considering $H_i = H_{i-1} \cup \{h_i\}$, it needs to make sure that the new fractional allocation restricted to the houses in $H_{i-1}$ remains the same (due to Lemma \ref{lem:eq-ass}).
	For the preference to be optimal with respect to $H_i$, the algorithm needs to maximize the fractional allocation of $h_i$ to agent 1 under the previous constraint.

	Our algorithm will compute a canonical DL best response that has several additional properties.

	\begin{definition}
	 A preference $L_i$ for $H_i$ is \emph{no-$0$} if $L_i$ contains no house $h$ with $PS1(L_i,h)=0$.
	\end{definition}

	\noindent
	Any DL best response for agent $1$ for $H_i$ can be converted into a no-$0$ DL best response by removing the houses for which agent 1 obtains a fraction of $0$.

	\begin{definition}
	 For a no-$0$ preference $L_i$ for $H_i$,
	 the \emph{stingy ordering} for a position $j$
	 is determined by running the PS rule with the preference $L_i(1) \oplus \dots \oplus L_i(j-1)$ for agent 1 where $\oplus$ denotes concatenation.
	 It orders the houses from $\bigcup_{k=j}^{|L_i|} L_i(k)$ by increasing eating start times, and when 2 houses $h,h'$ have the same eating start time, we order $h$ before $h'$ iff $h\succ_1 h'$.
	\end{definition}

	Intuitively, houses occurring early in this ordering are the most threatened by the other agents at the time point when agent 1 comes to position $j$.
	The following definition takes into account that the eating start times of later houses may change depending on agent 1's ordering of earlier houses.

	\begin{definition}
	 A preference $L_i$ for $H_i$ is \emph{stingy} if it is a no-$0$ DL best response for agent 1 on $H_i$, and
	 for every $j\in\{1,\dots,i\}$, $L_i(j)$ is the first house in the stingy ordering for this position such that there exists a DL best response starting with $L_i(1) \oplus \dots \oplus L_i(j)$.
	\end{definition}

	\noindent
	We note that, due to Lemma \ref{lem:eq-ass}, there is a unique stingy preference for each $H_i$.

	\begin{example}
	 Consider the following assignment problem.
	 \begin{align*}
		\succ_1: h_1, h_2, h_3, h_4, h_5, h_6&&
		\succ_2: h_3, h_6, h_4, h_5, h_1, h_2
	 \end{align*}
	 The preferences $h_3, h_1, h_4, h_2$ and $h_3, h_2, h_4, h_1$ are both no-$0$ DL best responses for agent 1 with respect to $H_4$, allocating $h_1 (1), h_2 (1), h_3 (1/2), h_4 (1/2)$ to agent 1. When running the PS rule with $h_3$ as the preference list, $h_4$'s eating start time comes first among $\{h_1,h_2,h_4\}$. However, there is no DL best response for $H_4$ starting with $h_3,h_4$. The next house in the stingy ordering is $h_1$. The preference $h_3, h_1, h_4, h_2$ is the stingy preference for $H_4$.
	\end{example}

	\noindent
	The next lemma shows that when agent 1 receives a house partially (a fraction different from 0 and 1) in a DL best response, a stingy preference would not order a less preferred house before that house.

	\begin{lemma}\label{lem:before-half}
	 Let $L_{i}$ be a stingy preference for $H_i$.
	 Suppose there is a $h_j\in H_i$ such that $0<PS1(L_{i},h_j)<1$.
	 Then, $P \subseteq H_j$, where $L_i = P \oplus h_j \oplus S$.
	\end{lemma}
	\begin{proof}
	 For the sake of contradiction, assume $P$ contains a house $h_k$ such that $h_j \succ_1 h_k$ (i.e., $j<k$).
	 Let $K$ denote all houses $h_k$ in $P$ such that $h_j \succ_1 h_k$.
	 Since $L_i$ is no-$0$, $PS1(L_i,h_k)>0$ for all $h_k \in K$.
	 But then, removing the houses in $K$ from $L_i$ gives a preference that is strictly DL preferred to $L_i$ since this increases agent 1's share of $h_j$ while only the shares of less preferred houses decrease.
	 This contradicts $L_i$ being a DL best response for $H_i$, and therefore proves the lemma.
	\end{proof}

	\noindent
	The next lemma shows how the houses allocated completely to agent 1 are ordered in a stingy preference.

	\begin{lemma}\label{lem:ones-same-order}
	 Let $L_i$ be a stingy preference for $H_i$.
	 If $h_j,h_k \in H_i$ are two houses such that
	 $PS1(L_i,h_j)=PS1(L_i,h_k)=1$, with $L_i = P\oplus h_j \oplus M \oplus h_k \oplus S$, then either the eating start time of $h_j$ is smaller than $h_k$'s eating start time when agent 1 reports $P$, or it is the same and $h_j \succ_1 h_k$.
	\end{lemma}
	\begin{proof}
	 Suppose not.
	 But then, $L_i$ is not stingy since swapping $h_j$ and $h_k$ in $L_i$ gives the same fractional allocation to agent 1.
	\end{proof}

	\noindent
	We now show that when iterating from a set of houses $H_{i-1}$ to $H_i$, the previous solution can be reused up to the last house that agent 1 receives partially.

	\begin{lemma}\label{lem:same-prefix}
	 Let $L_{i-1}$ and $L_i$ be stingy preferences for $H_{i-1}$ and $H_i$, respectively.
	 Suppose there is a $h\in H_{i-1}$ such that $0<PS1(L_{i-1},h)<1$.
	 Then the prefixes of $L_{i-1}$ and $L_i$ coincide up to $h$.
	\end{lemma}
	\begin{proof}
	 Suppose not.
	 By Lemma \ref{lem:eq-ass}, $PS1(L_i,h) = PS1(L_{i-1},h)$.
	 Let $P_{i-1}=P_i$ denote a maximum common prefix of $L_{i-1}$ and $L_i$, and write
	     $L_{i-1} = P_{i-1} \oplus x_{i-1} \oplus M_{i-1} \oplus h \oplus S_{i-1}$ and
	     $L_{i}   = P_{i}   \oplus x_{i}   \oplus M_{i}   \oplus h \oplus S_{i}$.
	 By Lemma \ref{lem:before-half}, $h \succ_1 h_i$, and therefore, $h_i\in S_i$.
	 Since $L_{i-1}$ and $L_i$ are no-$0$, we have that $PS1(L_{i-1},x_{i-1})>0$ and $PS1(L_{i},x_{i})>0$.
	 Now, if $PS1(L_{i-1},x_{i-1})<1$, then since at least one other agent eats $x_{i-1}$ concurrently with agent 1 when he reports $L_{i-1}$, he loses a non-zero fraction of $x_{i-1}$ when instead he reports $L_i$ and eats $x_i$ after having exhausted $P_i$, we have that $PS1(L_{i},x_{i-1}) < PS1(L_{i-1},x_{i-1})$, a contradiction to Lemma \ref{lem:eq-ass}.
	 Similarly, we obtain a contradiction when $PS1(L_{i},x_{i})<1$.
	 Therefore, $PS1(L_{i-1},x_{i-1})=PS1(L_{i},x_{i})=1$.
	 Now, by Lemma \ref{lem:eq-ass}, we also have that $PS1(L_{i-1},x_{i-1})=PS1(L_{i},x_{i})=1$.
	 But only one of $x_i,x_{i-1}$ can come earlier in the stingy ordering. The other one contradicts Lemma \ref{lem:ones-same-order}.
	\end{proof}

		\begin{algorithm}[tb]
			\caption{DL best response for $n$ agents}
			\label{algo:DL-BR}
			\renewcommand{\algorithmicrequire}{\wordbox[l]{\textbf{Input}:}{\textbf{Output}:}}
			\renewcommand{\algorithmicensure}{\wordbox[l]{\textbf{Output}:}{\textbf{Output}:}}
			\renewcommand{\algorithmiccomment}[1]{\hfill // #1}
			\footnotesize
			\begin{algorithmic}
				\REQUIRE $(N,H,\succ)$ 
				\ENSURE DL Best response of agent $1$
			\end{algorithmic}
			\algsetup{linenodelimiter=\,}
			\begin{algorithmic}[1]
				\STATE $L_1 \leftarrow h_1$ \COMMENT{Best response for agent 1 w.r.t. $H_1 = \{h_1\}$}
				\FOR[Compute a best response w.r.t. $H_2, \dots, H_n$]{$i=2$ to $n$}
				  \STATE $p \leftarrow 0$ 	
				  \IF{$\exists q\in \{1,\dots,i-1\}$ such that $0<PS1(L_{i-1},L_{i-1}(q))<1$}
				    \STATE $p \leftarrow \max\{q\in \{1,\dots,i-1\} : 0<PS1(L_{i-1},L_{i-1}(q))<1\}$
				  \ENDIF
				  \FOR[New house $h_i$ inserted after position $p$]{$q \leftarrow p+1$ to $|L_i|+1$}
				    \STATE $L_i^q \leftarrow L_{i-1}(1) \oplus \dots \oplus L_{i-1}(q-1) \oplus h_i$
				    \WHILE[Complete the list according to the stingy ordering]{$|L_i^q| \le |L_{i-1}|$}
				      \STATE $est \leftarrow$ \textbf{EST}$(N,H,(L_i^q, \succ_2, \dots, \succ_n))$
				      \STATE $S \leftarrow \{h\in L_{i-1} \setminus L_i^q : est(h) \text{ is minimum}\}$
				      \STATE $h_s \leftarrow $ first house among $S$ in $\succ_1$
				      \STATE $L_i^q \leftarrow L_i^q \oplus h_s$
				    \ENDWHILE
				    \IF{$PS1(L_i^q,h_i)=0$}
				      \STATE $L_i^q \leftarrow L_{i-1}$
				    \ENDIF
				  \ENDFOR				  
				  \STATE $q\leftarrow p$
				  \COMMENT{Determine which $L_i^q$ is stingy}
				  \STATE $\mathsf{worse}[p-1] \leftarrow \mathsf{true}$
				  \STATE $\mathsf{finished} \leftarrow \mathsf{false}$
				  \WHILE{$\mathsf{finished}=\mathsf{false}$}
				    \IF{$\exists h\in H_{i-1}$ such that $PS1(L_i^q,h) \ne PS1(L_{i-1},h)$}
				      \STATE $\mathsf{worse}[q] \leftarrow \mathsf{true}$
				      \STATE $q \leftarrow q+1$
				    \ELSE
				      \STATE $\mathsf{worse}[q] \leftarrow \mathsf{false}$
				      \IF{$PS1(L_i^q,h_1)>0$ \AND $PS1(L_i^q,h_1)<1$}
				        \IF{$\mathsf{worse}[q-1] = \mathsf{false}$}
				          \STATE $q \leftarrow q-1$
				        \ENDIF
				        \STATE $\mathsf{finished} \leftarrow \mathsf{true}$
				      \ELSIF{$PS1(L_i^q,h_1)=1$}
				        \STATE $est \leftarrow$ \textbf{EST}$(N,H,(L_i^q(1) \oplus \dots \oplus L_i^q(q-1), \succ_2, \dots, \succ_n))$
				        \IF{$\exists h\in \{L_i^q(q+1), \dots, L_i^q(|L_i^q|)\}$ such that $est(h)\le est(h_i)$}
				          \STATE $q \leftarrow q+1$
				        \ELSE
				          \STATE $\mathsf{finished}=\mathsf{true}$
				        \ENDIF
				      \ENDIF
				    \ENDIF
				  \ENDWHILE
				  \STATE $L_i \leftarrow L_i^q$
				\ENDFOR
				\RETURN $L_n$
			\end{algorithmic}
			\end{algorithm}
	
	\noindent
	We are now ready to describe how to obtain $L_i$ from $L_{i-1}$. See Algorithm~\ref{algo:DL-BR} for the pseudocode.
	The subroutine \textbf{EST}$(N,H,\succ)$ executes the PS rule for $(N,H,\succ)$ and for each item, records the first time point where some agent starts eating it. It returns the eating start times $est(h)$ for each house $h\in H$.
	
	Let $p$ be the last position in $L_{i-1}$ such that the house $L_{i-1}(p)$ is partially allocated to agent 1.
	In case agent 1 receives no house partially, set $p:=0$ and interpret $L_{i-1}(p)$ as an imaginary house before the first house of $L_{i-1}$.
	By Lemma \ref{lem:same-prefix}, we have that $L_{i-1}(s) = L_i(s)$ for all $s\le p$.
	By Lemma \ref{lem:eq-ass}, we have that the fractional assignment resulting from $L_i$ must wholly allocate all houses $L_{i-1}(p+1), \dots, L_{i-1}(|L_{i-1}|)$ to agent 1, and allocate a share of $0$ to all houses in $H_{i-1}\setminus L_{i-1}$.

	It remains to find the right ordering for $\{L_{i-1}(s): p+1\le s\le |L_{i-1}|\} \cup \{h_i\}$.
	By Lemmas \ref{lem:before-half} and \ref{lem:ones-same-order}, the prefixes of $L_{i-1}$ and $L_i$ coincide up to $h$.
	We will describe in the next paragraph how to determine the position $q$ where $h_i$ should be inserted.
	Having determined this position one may then need to re-order the subsequent houses.
	This is because inserting $h_i$ in the list may change the eating start times of the subsequent houses.
	This leads us to the following insertion procedure.
	The list $L_i^q$ obtained from $L_{i-1}$ by inserting $h_i$ at position $q$, with $p<q\le|L_i|+1$, is determined as follows.
	Start with $L_i^q := L_{i-1}(1) \oplus \dots \oplus L_{i-1}(q-1) \oplus h_i$.
	While $|L_i^q|\le |L_{i-1}|$, we append to the end of $L_i^q$ the first house among $L_{i-1}\setminus L_i^q$ in the stingy ordering for this position.
	After the while-loop terminates, run the PS rule for the resulting list $L_i^q$.
	In case we obtain that $PS1(L_i^q,h_i)=0$, we remove $h_i$ again from this list (and actually obtain $L_i^q=L_{i-1}$).

	The position $q$ where $h_i$ is inserted is determined as follows.
	Start with $q:=p$.
	We have an array $\mathsf{worse}$ keeping track of whether the lists $L_i^p , \dots, L_i^{i}$ produce a worse
	outcome for agent 1 than the list $L_{i-1}$.
	Set $\mathsf{worse}[p-1]:=\mathsf{true}$.
	As long as the list $L_i$ has not been determined, proceed as follows.
	Obtain $L_i^q$ from $L_{i-1}$ by inserting $h_i$ at position $q$, as described earlier.
	Consider the allocation of agent 1 when he reports $L_i^q$.
	If this allocation is not the same for the houses in $H_{i-1}$ as when reporting $L_{i-1}$, then set $\mathsf{worse}[q] := \mathsf{true}$, otherwise set $\mathsf{worse}[q] := \mathsf{false}$.
	If $\mathsf{worse}[q]$, then increment $q$.
	This is because, by Lemma \ref{lem:eq-ass}, this preference would not be a DL best response with respect to $H_i$.
	Otherwise, if $0<PS1(L_i^q,h_i)<1$, then we can determine $h_i$'s position. If $\mathsf{worse}[q-1]$, then set $L_i := L_i^q$, otherwise set $L_i:= L_i^{q-1}$.
	This position for $h_i$ is optimal since moving $h_i$ later in the list would decrease its share to agent 1.
	Otherwise, we have that $\mathsf{worse}[q] = \mathsf{false}$ and $PS1(L_i^q,h_i) \in \{0,1\}$.
	This will be the share agent 1 receives of $h_i$.
	If $PS1(L_i^q,h_i) = 0$, then set $L_i:=L_{i-1}$.
	Otherwise ($PS1(L_i^q,h_i) = 1$), it still remains to check whether the current position for $h_i$ gives a stingy preference.
	For this, run the PS rule with the preference $L_i^q(1) \oplus \dots \oplus L_i^q(q-1)$ for agent 1. If $h_i$'s eating start time is smaller than the eating start time of each house $L_i^q(r)$ with $r>q$, then set $L_i := L_i^q$, otherwise increment $q$.

	Thus, given $L_{i-1}$, the preference $L_i$ can be computed by executing the PS rule $O(m)$ times.
	The DL best response computed by the algorithm is $L_m$.
	Since the PS rule can be implemented to run in linear time $O(nm)$, the running time of this DL best response algorithm is $O(nm^3)$.

	\begin{theorem}
	 \textsc{DLBestResponse} can be solved in $O(nm^3)$ time.
	\end{theorem}

	\begin{example}
	 Consider the following instance. 
	 \begin{align*}
		\succ_1: h_1, h_2, h_3, h_4, h_5, h_6, h_7, h_8, h_9, h_{10}\\
		\succ_2: h_8, h_3, h_5, h_2, h_{10}, h_1, h_6, h_7, h_4, h_9\\
		\succ_3: h_9, h_4, h_7, h_1, h_2, h_6, h_5, h_3, h_8, h_{10}
	 \end{align*}
	 After having computed $L_2=h_1, h_2$, the algorithm is now to consider $H_3$. Since $PS1(L_2,h_1)=PS1(L_2,h_2)=1$, the algorithm first considers $L_3^1 = h_3, h_2, h_1$. Note that $h_1$ and $h_2$ have been swapped with respect to $L_2$ since agent 2 starts eating $h_2$ before agent 3 starts eating $h_1$ when agent 1 reports the preference list consisting of only $h_3$. It turns out that $PS1(L_3^1,h_1)=PS1(L_3^1,h_2)=PS1(L_3^1,h_3)=1$. Thus, $\mathsf{worse}[1] = \textsf{false}$. Since $h_3$ does not come first in the stingy ordering, the algorithm needs to verify whether moving $h_3$ later will still give a DL best response with respect to $H_3$.
	 It then considers $L_3^2 = h_1, h_3, h_2$. However, this allocates only half of $h_3$ to agent 1, implying $\mathsf{worse}[2] = \mathsf{true}$. Since $\mathsf{worse}[1] = \mathsf{false}$, the algorithm sets $L_3 = L_3^1$.
	 The DL best response computed by the algorithm is $L_{10}=h_3, h_2, h_1, h_6$.
	\end{example}
	\begin{figure}[htb]
		 \centering
		     \[
		             \begin{tikzpicture}[scale=0.25]
		                 \centering
		                 \draw[-] (0,0) -- (0,8);
		                 \draw[-] (0,0) -- (20,0);
		                 \draw[-] (20,8) -- (20,0);

	\draw[-] (0,8) -- (20,8);
		\draw[-] (0,2) -- (20,2);
		\draw[-] (0,4) -- (20,4);
		\draw[-] (20,0) -- (20,6);


		\draw[-] (0,6) -- (20,6);

	\draw[-] (20/10,0) -- (20/10,8);
	\draw[-] (60/10,0) -- (60/10,8);
	\draw[-] (80/10,0) -- (80/10,8);
	\draw[-] (120/10,0) -- (120/10,8);
	\draw[-] (140/10,0) -- (140/10,8);
	\draw[-] (180/10,0) -- (180/10,8);
	\draw[-] (200/10,0) -- (200/10,8);

		                                        \draw (0,-.8) node(c){\small $0$};
		\draw (20/10,-.8) node(c){\small $\frac{1}{3}$};
		\draw (60/10,-.8) node(c){\small ${1}$};
		\draw (80/10,-.8) node(c){\small $\frac{4}{3}$};
		\draw (120/10,-.8) node(c){\small $2$};
		\draw (140/10,-.8) node(c){\small $\frac{7}{3}$};
			\draw (180/10,-.8) node(c){\small $\frac{9}{3}$};
					\draw (200/10,-.8) node(c){\small $\frac{10}{3}$};

		 \draw(-2.5,8) node(z){\small Agent $1$};
		    \draw(-2.5,6) node(z){\small Agent $2$};
		                 \draw(-2.5,4) node(z){\small Agent $3$};
		                 \draw(-2.5,2) node(z){\small Agent $4$};

		 \draw(10/10,8.4) node(z){\small $h_1$};
			 \draw(10/10,6.4) node(z){\small $h_5$};
				 \draw(10/10,4.4) node(z){\small $h_1$};
		 \draw(10/10,2.4) node(z){\small $h_1$};

		\draw(40/10,8.4) node(z){\small $h_2$};
		\draw(40/10,6.4) node(z){\small $h_5$};
		\draw(40/10,4.4) node(z){\small $h_8$};
	\draw(40/10,2.4) node(z){\small $h_{11}$};

		\draw(70/10,8.4) node(z){\small $h_2$};
		\draw(70/10,6.4) node(z){\small $h_6$};
		\draw(70/10,4.4) node(z){\small $h_8$};
	\draw(70/10,2.4) node(z){\small $h_{11}$};

		\draw(100/10,8.4) node(z){\small $h_4$};
		\draw(100/10,6.4) node(z){\small $h_6$};
		\draw(100/10,4.4) node(z){\small $h_9$};
	\draw(100/10,2.4) node(z){\small $h_{12}$};

		\draw(130/10,8.4) node(z){\small $h_4$};
		\draw(130/10,6.4) node(z){\small $h_7$};
		\draw(130/10,4.4) node(z){\small $h_9$};
	\draw(130/10,2.4) node(z){\small $h_{12}$};

		\draw(160/10,8.4) node(z){\small $h_3$};
		\draw(160/10,6.4) node(z){\small $h_7$};
		\draw(160/10,4.4) node(z){\small $h_{10}$};
	\draw(160/10,2.4) node(z){\small $h_{13}$};

		\draw(190/10,8.4) node(z){\small $h_3$};
		\draw(190/10,6.4) node(z){\small \textcolor{gray}{$h_4$}};
		\draw(190/10,4.4) node(z){\small $h_{10}$};
	\draw(190/10,2.4) node(z){\small $h_{13}$};

		%
		%
		%
		%
		%
		%
		%
		%
		%
		%
		  \end{tikzpicture}
		\vspace{-1em}
		         \]\begin{align*}
					\succ_1: &\quad h_1,h_2,h_3,h_4, \ldots&
					\succ_2: &\quad h_5,h_6,h_7,h_2, h_4,h_{14} \ldots\\
					\succ_3: &\quad h_1,h_8,h_9,h_{10}, h_3, \ldots&
					\succ_4: &\quad h_1,h_{11},h_{12},h_{13},  \ldots
				\end{align*}

		\caption{Illustration of constructing a DL best response for agent $1$ for the preference profile specified above.}
		\label{figure:serge-counter-example}
		 \end{figure}


\begin{example}
	Figure~\ref{figure:serge-counter-example} depicts how the DL best response of agent $1$ looks like. After $h_1$ is inserted, the starting eating time $h_3$ is before $h_4$. But after $h_2$ is inserted in to form $L_2$, then the starting eating time of $h_4$ comes before $h_3$ because agent $2$ won't be able to eat $h_2$. After $h_4$ is inserted to build $L_4$, it turns out that agent $2$ will not be able to eat $h_4$ at all. That is why $h_2$ is shaded in the eating line of agent $2$ because it will already be eaten by the time agent $2$ considers eating it at time $10/3$.

	The DL optimal best response algorithm carefully builds up the DL optimal preferences list while ensuring it is stingy.
\end{example}

	\noindent
	We note that a DL best response is also an SD best response. A best response was defined as a response that is not dominated. Hence a DL-best response is one which no other response DL-dominates. This means that no other response SD-dominates (as DL is a refinement of SD) it. Hence, a DL best response is also a SD best response.
	One may wonder whether an algorithm to compute the DL best response also provides us with an algorithm to compute an EU best response. However, a DL best response may not be an EU best response for three or more agents.
		Consider the preference profile in Example~\ref{example:PS}. Since the number of houses is equal to the number of agents, reporting the truthful preference is a DL best response~\citep{ScVa12a}. However, we have shown a different preference for agent 1 where he may obtain higher utility.


	\section{Expected utility best response}

	In this section
	we present an algorithm to compute an EU best response for two agents for the PS rule.
	First, we reveal a tight connection between a well-known mechanism for sequential allocation of indivisible houses and the $\ps$ mechanism (Section~\ref{s:connection_alt_ps}). Then we demonstrate how the expected utility best response algorithm for the sequential allocation of indivisible houses proposed by~\citet{KoCh71a} can be used to build a best response for the $\ps$ algorithm (Section~\ref{s:eu_br_alt_ps}).

		\subsection{A connection between allocation mechanisms for divisible and indivisible houses}
		\label{s:connection_alt_ps}

	We can obtain the same allocation given by the $\ps$ algorithm using the \emph{alternation policy}, which is a simple mechanism for dividing discrete houses between agents. The alternating policy lets the agents take turns in picking the house that they value most: the first agent takes his most preferred house, then the second agent takes his most preferred house from the remaining houses, and so on. We use the notation $1212\ldots$ to denote the alternation policy. To obtain the allocation of the $\ps$ algorithm using the alternation policy we split our houses into halves and treat them as indivisible houses and adjust agents' preferences over these halves in a natural way.

		Recall that  $H = \{h_1,\ldots,h_m\}$
		is the set of houses.
		Assume $\pref_2: h_1, \ldots, h_m$
	 and the preference of
		agent 1 is a permutation of $h_1, \ldots, h_m$ as follows
		$\pref_1: h_{\pi(1)},\ldots, h_{\pi(m)}.$
		We denote $\preffn{i}{k}$ the $k$th preferred house of the agent $i$,
		and by  $\prefinvfn{1}{h_i}$ we mean the position of $h_i$ in $\pref_1$.

		We split each house $h_i$, $i=1,\ldots,m$, into halves and treat
		these halves as indivisible houses. Given $h_i$, we say
		that $h_i^1$ and $h_i^2$ are two halves of $h_i$.
		Given the set of houses $H$, we denote $\hclone$ the
		set of all halves of all houses in $H$, so that
		$\hclone = \{h_1^1, h_1^2, \ldots, h_m^1, h_m^2\}$.
		Given $\pref_1$ and $\pref_2$, we introduce profiles $\prefcloned_1$ and $\prefcloned_2$ that are obtained
		by straightforward splitting of houses into halves in $\pref_1$ and $\pref_2$:
		$\prefcloned_1 = h_{\pi(1)}^1, h_{\pi(1)}^2, \ldots, h_{\pi(m)}^1, h_{\pi(m)}^2$
		and
		$\prefcloned_2 = h_1^1, h_1^2, \ldots, h_m^1, h_m^2$.
		We call this transformation the \emph{order-preserving bisection}.





		\begin{definition}
		Let $\pref_s$ be a preference over a subset of half-houses $S \subseteq \hclone$.
		The preference $\pref_s$ has the \emph{consecutivity} property if and only if $\prefinvfn{s}{h_i^1} + 1 = \prefinvfn{s}{h_i^2}$
		for all pairs $h_i^1, h_i^2 \in S$. In other words, all half-houses of the same house
		are ranked consecutively in $\pref_s$.
		\end{definition}



		The preference $\pref_s =  h_1^1,h_2^1, h_2^2, h_3^1 ,h_3^2$ has the consecutivity property over the set $S = \{h_1^1, h_2^1, h_2^2, h_3^1, h_3^2\}$,
		while $\pref_s = h_1^1,h_2^1,h_3^1, h_2^2 ,h_3^2$ does not since $h_2^1 \succ h_3^1 \succ h_2^2$.
		We observe that $\prefcloned_1$ and $\prefcloned_2$ that are obtained from $\pref_1$ and $\pref_2$ using
	the order-preserving bisection, respectively, have the consecutivity property.

	Next, we define the \emph{order-preserving join} operation. It is the reverse operation for the order-preserving bisection.
	Given a preference  $\prefcloned_s$ of the order-preserving join operation merges all halve houses that are ordered
	consecutively into a single house and leaves the other houses unchanged.  
	Applying the order-preserving join to $\prefcloned_s = h_1^1,h_1^2,h_2^1,h_3^1,h_3^2,h_2^2.$ gives $\pref_s =h_1,h_2^1,h_3,h_2^2$.

	%


		Next, we show the main result of this section. The outcome of the alternation
		policy
	over
	 $\prefcloned_1$  and $\prefcloned_2$ is identical to the outcome of $\ps$ over $\pref_1$ and $\pref_2$,
		where $\prefcloned_1$ and $\prefcloned_2$ are obtained by the order-preserving bisection from $\pref_1$ and $\pref_2$.
		In the alternation policy $12,\ldots, 12$ we call a pair of consecutive steps $12$ \emph{a round}.

		\begin{lemma}\label{l:ps2alt}
		The allocation obtained by the $\ps$ algorithm
		over the preferences $\pref_1$ and $\pref_2$ of length $m$ is the same as
		the allocation obtained by the alternation policy of length $2m$
		over  the preferences $\prefcloned_1$ and $\prefcloned_2$.
		\end{lemma}
		\begin{proof}
		The proof is by induction on the number of steps of the PS rule. A step in the PS rule starts when agent $1$ starts eating a house and finishes when agent 1 finishes eating that house.
		For the base case, at time point $0$, both the PS algorithm and the alternation policy have not allocated a house to any agent.

		Suppose the statement holds for $i-1$ steps of the PS rule, where $i\ge 1$.
		If both agents have the same most preferred house $h_k$ among the remaining houses,
		then each of them gets half of this house in the PS rule.
		Consider the next round of the alternation policy:
		agent $1$ gets a half of $h_k$, $h_k^1$ and
		agent $2$ gets the other half of $h_k$, $h_k^2$.
		Hence, the allocation is the same.

		If the most preferred houses of the two agents are different, say the most preferred house among the remaining houses of agent 1 is $h_j$, and the most preferred of agent 2 is $h_k$,
		then agent 1 completely receives house $h_j$ and agent 2 completely receives house $h_k$ in step $i$ of the PS rule.
		In the alternation policy, agent 1 gets $h_j^1, h_j^2$
		and agent $2$ gets $h_k^1, h_k^2$ in the next two rounds.
		Hence, the allocation is the same.
	   \end{proof}
	%

		\begin{example}\label{exm:eqalt2ps}
		Consider two agents with preferences $\pref_1 = h_5, h_6, h_1, h_3, h_4, h_2$
		and  $\pref_2 = h_1, h_2, h_3, h_4, h_5, h_6$.
		The allocation obtained by the $\ps$ algorithm over $\pref_1$ and $\pref_2$ is
	$
	PS(\pref_1,\pref_2)=\begin{pmatrix}
		   0 & 0 & 1/2 & 1/2 & 1 & 1\\
		   1 & 1 & 1/2 & 1/2 & 0& 0
	 \end{pmatrix}.
	$
	The identical allocation given by 
	the alternation policy with 
	$\prefcloned_1$ and $\prefcloned_2$ is
		$
	\small
		\begin{array}{|cccccc|}
		\hline
		\multicolumn{6}{|c|}{\mathrm{Rounds}} \\
	1 &   2 &  3 &  4  & 5   & 6 \\
		\hline
		\hline
		  h_5^1   &  h_5^2 &  h_6^1 & h_6^2  &   h_3^1&  h_4^1\\
		   h_1^1 &  h_1^2 &   h_2^1  &   h_2^2   & h_3^2   & h_4^2  \\
		\hline
		\end{array}.
		 $
		\end{example}

	\subsection{Computing an EU best response}
		\label{s:eu_br_alt_ps}

	In this section we present an algorithm to compute an expected utility best response for the $\ps$ mechanism.
	First, we recap our settings. We are given two agents $1$ and $2$ with profiles
	$\pref_1$ and $\pref_2$, respectively, over houses in $H$.
	We assume that  agent 1 plays strategically and agent 2 plays truthfully.
	The goal is to find an expected utility best response for agent 1 for the $\ps$ rule.
	To do so, we reuse an EU best response
	for the alternation policy over split houses,
	$\prefcloned_1$ and $\prefcloned_2$.
	Our algorithm is based on the following lemma.
	Let $\prefclonedbest_1$ be an expected utility best response for agent 1
	to $\prefcloned_2$ for the alternation policy.

		\begin{lemma}\label{l:altbr2psbr}
	Suppose $\prefclonedbest_1$ has the consecutivity property.
		Then, $\prefbest_1$, obtained by the order-preserving join
		from $\prefclonedbest_1$ is an EU best response to $\pref_2$.
		\end{lemma}
		\begin{proof}
		The proof is by contradiction.
		Suppose, ${\prefbest_1}'$ is EU preferred to $\prefbest_1$. We transform
		${\prefbest_{1}}'$ into ${\prefclonedbest_{1}}'$ using the order-preserving bisection.
		By Lemma~\ref{l:ps2alt}, if we run the alternation policy over
		${\prefclonedbest_{1}}'$ and $\prefcloned_2$, the agents get the same allocation
		as by running $\ps$. Hence, $\prefclonedbest_1$ is not the best response to
		$\prefcloned_2$. This leads to a contradiction.
		\end{proof}

	    \noindent
		Lemma~\ref{l:altbr2psbr} suggests a straightforward way
		to  compute agent 1's  best response $\prefbest_1$ for the $\ps$
		algorithm.
	We run Kohler and Chandrasekaran's algorithm that finds a best response
	for the alternation policy given agents' preferences $\prefcloned_1$ and $\prefcloned_2$.
	If $\prefclonedbest_1$ has the consecutivity property then we can use
	the order-preserving join  to obtain $\prefbest_1$ which is the expected utility   best response
	to $\pref_2$ in $\ps$ by Lemma~\ref{l:altbr2psbr}.
	The main problem with this approach is that the algorithm of \citet{KoCh71a} may return  $\prefclonedbest_1$ that does not have the  consecutivity property (we provide
	such an example in the full
	report).
	However, we show in Algorithm~\ref{algo:2agent-EU-BR} that we can always find another expected utility  best response ${\prefclonedbest_1}'$ that has the consecutivity property. We need to delay the allocation of some half-houses that agent 1 gets. 
		\begin{algorithm}[tb]
			  \caption{EU best response for $2$ agents}
			  \label{algo:2agent-EU-BR}	\renewcommand{\algorithmicrequire}{\wordbox[l]{\textbf{Input}:}{\textbf{Output}:}}
			 \renewcommand{\algorithmicensure}{\wordbox[l]{\textbf{Output}:}{\textbf{Output}:}}
				\footnotesize
			\begin{algorithmic}
				\REQUIRE $(\{1,2\},H,(\succ_1,\succ_2))$ where $\succ_2: h_1,\succ_2 \ldots,\succ_2 h_m$
				\ENSURE Best response $\prefbest_1$ of agent $1$
			\end{algorithmic}
			\algsetup{linenodelimiter=\,}
			  \begin{algorithmic}[1]
		\STATE Construct order-preserving bisection of $\succ_1$ and $\succ_2$: $\prefcloned_1$ and $\prefcloned_2$.
		\STATE Run Kohler and Chandrasekaran's algorithm for two agents with preferences $\prefcloned_1$ and $\prefcloned_2$ and use the alternation policy. We obtain $\prefclonedbest_1$.
		\IF{$\prefclonedbest_1$ does not satisfy consecutivity} \label{alg:delaybegin}
				\STATE Let $W$ be a set of half-houses such that agent $1$ gets a half-house $h^1$ but does not get $h^2$.
	            \FOR {all $h^1 \in W$} \label{alg:delaybegin_loop}
	            \STATE Suppose agent 2 gets  $h^2$ in the $i$th round.
	            \STATE Move $h^1$ at the $i$th position in $\prefclonedbest_1$.
	\ENDFOR \label{alg:delayend_loop}
	            \FOR {all $h^1 \in W$} \label{alg:conend_loop}
	            \STATE Rank $h^2$ right after $h^1$ in $\prefclonedbest_1$.
	\ENDFOR\label{alg:conbegin_loop}

	\ENDIF \label{alg:delayend}
	\STATE Use  order-preserving join to obtain $\prefbest_1$ from $\prefclonedbest_1$
	\RETURN $\prefbest_1$
					 \end{algorithmic}
			\end{algorithm}
	%
	The modifications of the best response $\prefclonedbest_1$  in lines~\ref{alg:delaybegin}--\ref{alg:delayend}
	produce another\emph{ best response } that has the consecutivity property for agent 1.
	A detailed description of the algorithm from~\citep{KoCh71a}  and a proof of correctness of Algorithm~\ref{algo:2agent-EU-BR}  can be found in
	the full report.

	\begin{remark}\label{remark:dl-eu}
	The EU best response algorithm is independent of particular utilities and holds for any utilities consistent with the ordinal preferences. Since PS for two agents only involves fractions $0,\frac{1}{2}$, and $1$,  a DL best response is also equivalent to an EU best response.
	Hence we have proved that the DL best response algorithm in Section~\ref{sec:dl} is also an EU best response algorithm for the case of two agents.
	\end{remark}

	\section{Nash dynamics and equilibrium}



	In contrast to the previous sections where a single agent is strategic, we consider the setting when all the agents are strategic. We first prove that for expected utility best responses,
	the preference profile of the agents can cycle when agents have Borda utilities.  This means
	that it is possible that self interested agents, acting unilaterally,
	may never stop reacting.

	\begin{theorem}\label{th:cycle}
	With 3 agents and 6 items where agents have Borda utilities, 
	a series of expected utility best responses by the agents can lead to a cycle
	in the profile.
	\end{theorem}

	Using a computer program we have found a sequence of best response that cycle.

Checking the existence of a preference profile that is in Nash equilibrium appears to be a challenging problem. The naive way of checking existence of Nash equilibrium requires going through $O({m!}^n)$ profiles, which is super-polynomial even when $n=O(1)$ or $m=O(1)$. 
	Although computing a Nash equilibrium is a challenging problem, we show that at least one (pure) Nash equilibrium is guaranteed to exist for any number of houses, any number of agents, and any preference relation over fractional allocations.\footnote{We already know from Nash's original result that a \emph{mixed} Nash equilibrium exists for any game.} The proof relies on showing that the PS rule can be modelled as a perfect information extensive form game. 	

	\begin{theorem}
	A pure Nash equilibrium is guaranteed to exist under the PS rule for any number of agents and houses, and for any relation between allocations.
	\end{theorem}
	\begin{proof}[Sketch]

		Let $t^0,\dots, t^k$ be the $k+1$ different time steps in the PS algorithm. Let $g=\text{GCD}(\{t^{i+1}-t^i\midd i\in \{0,\ldots, k-1\}\})$ where GCD denotes the greatest common divisor. The time interval length $g$ is small enough such that the PS rule can be considered to have $m/g$ stages of duration $g$. Each stage can be viewed as having $n$ sub-stages so that in each stage, agent $i$ eats $g$ units of a house in sub-stage $i$ of a stage.
	In each sub-stage only one agent eats $g$ units of the most favoured house that is available. Hence we now view PS as consisting of a total of $mn/g$ sub-stages and the agents keep coming in order $1,2,\ldots, n$ to eat $g$ units of the most preferred house that is still available. If an agent ate $g$ units of a house in a previous sub-stage then it will eat $g$ units of the same house in the next sub-stage as long as the house has not been fully eaten.
	Consider a perfect information extensive form game tree.
	For a fixed reported preference profile, the PS rule unravels accordingly along a path starting at the root and ending at a leaf. Each level of the tree represents a sub-stage in which a certain agent has his turn to eat $g$ units of his most preferred available house. Note that there is a one-to-one correspondence between the paths in the tree and the ways the PS algorithm can be implemented, depending on the reported preference. 

	A subgame perfect Nash equilibrium is guaranteed to exist for such a game via backward induction:
	starting from the leaves and moving towards the root of the tree, the agent at the specific node chooses an action that maximizes his utility given the actions determined for the children of the node. The subgame perfect Nash equilibrium identifies at least one such path from a leaf to the root of the game. The path can be used to read out the most preferred house of each agent at each point. The information provided is sufficient to construct a preference profile that is in Nash equilibrium. Those houses that an agent did not eat at all can conveniently be placed at the end of the preference list. Such a preference profile is in Nash equilibrium. Hence, a pure Nash equilibrium exists under the PS rule.
	\end{proof}

	 We also know that DL-Nash equilibrium is an SD-Nash equilibrium because if there is an SD deviation, then it is also a DL deviation.
	Our argument for the existence of a Nash equilibrium is constructive. However, naively constructing the extensive form game and then computing a sub-game perfect Nash equilibrium requires exponential space and time. It is an open question whether a sub-game perfect Nash equilibrium or for that matter any Nash equilibrium preference profile can be computed in polynomial time. We can prove the following theorem for the ``threat profile'' whose construction is shown in Algorithm~\ref{algo:2agent-DL-Nash}.

	\begin{theorem}\label{th:threat}
		Under PS and for two agents, there exists a preference profile that is in DL-Nash equilibrium and results in the same assignment as the assignment based on the truthful preferences. Moreover, it can be computed in linear time.
	\end{theorem}
	\begin{proof}
		The proof is by induction over the length of the preference lists constructed.
		The main idea of the proof is that if both agents compete for the same house then they do not have an incentive to delay eating it. If the most preferred houses do not coincide, then both the agents get them with probability one but will not get them completely if they delay eating them.
		
	Let the original preferences of agent 1 and agent 2 be represented by lists $P_1$ and $P_2$. We present an algorithm to compute preferences $Q_1$ and $Q_2$ that are in DL-Nash equilibrium.
	Initialise $Q_1$ and $Q_2$ to empty lists. Now consider the maximal elements  $h$ from $P_1$ and $h'$ from $P_2$.
	Element $h$ is appended to the list $Q_1$ and $h'$ is appended to the list $Q_2$. At the same time  $h$ is deleted from $P_1$ and $h'$ is deleted from $P_2$. Now if $h\neq h'$, then $h'$ is appended to $Q_1$ and $h'$ is appended to $Q_2$. The process is repeated until $Q_1$ and $Q_2$ are complete lists and $P_1$ and $P_2$ are empty lists. The algorithm is described as Algorithm~\ref{algo:2agent-DL-Nash}.

	We now prove that $P_1$ is a DL best response against $P_2$ and $P_2$ is a DL best response against $P_1$. The proof is by induction over the length of the preference lists.
	For the first elements in the preference lists $P_1$ and $P_2$, if the elements coincide, then no agent has an incentive to put the element later in the list since the element is both agents' most preferred house. If the maximal elements do not coincide i.e. $h\neq h'$, then $1$ and $2$ get $h$ and $h'$ respectively with probability one. However they still need to express these houses as their most preferred houses because if they don't, they will not get the house with probability one. The reason is that $h$ is the next most preferred house after $h'$ for agent $2$ and $h'$ is the next most preferred house after $h$ for agent $1$. Agent $1$ has no incentive to change the position of $h'$ since $h'$ is taken by agent $2$ completely before agent $1$ can eat it. Similarly, agent $2$ has no incentive to change the position of $h$ since $h$ is taken by agent $1$ completely before agent $2$ can eat it.
	Now that the positions of $h$ and $h'$ have been completely fixed, we do not need to consider them and we reason in the same manner over the updated lists $P_1$ and $P_2$.\qed
		\end{proof}

	The desirable aspect of the threat profile is that since it results in the same assignment as the assignment based on the truthful preferences, the resultant assignment satisfies all the desirable properties of the PS outcome with respect to the original preferences. Due to Remark~\ref{remark:dl-eu}, we get the following corollary.

	\begin{corollary}
		Under PS and for 2 agents, there exists a preference profile that is Nash equilibrium for any utilities consistent with the ordinal preferences. Moreover it can be computed in linear time.
	\end{corollary}


	\begin{algorithm}[tbh]
		  \caption{Threat profile DL-Nash equilibrium for $2$ agents (which also is an EU Nash equilibrium)}
		  \label{algo:2agent-DL-Nash}
	\small	\renewcommand{\algorithmicrequire}{\wordbox[l]{\textbf{Input}:}{\textbf{Output}:}}
		 \renewcommand{\algorithmicensure}{\wordbox[l]{\textbf{Output}:}{\textbf{Output}:}}
		\begin{algorithmic}
			\REQUIRE $(\{1,2\},H,(\succ_1,\succ_2))$
			\ENSURE The \emph{``threat profile''} $(Q_1,Q_2)$ where $Q_i$ is the preference list of agent $i$ for $i\in \{1,2\}$.
		\end{algorithmic}
				\footnotesize
		\algsetup{linenodelimiter=\,}
		  \begin{algorithmic}[1]
	\STATE Let $P_i$ be the preference list of agent $i\in \{1,2\}$
	\STATE
	Initialise $Q_1$ and $Q_2$ to empty lists.
	\WHILE{$P_1$ and $P_2$ are not empty}
	\STATE Let $h= \text{first}(P_1)$ and $h'=\text{first}(P_2)$
	\STATE Append $h$ to $Q_1$; Append $h'$ to $Q_2$
	\STATE Delete $h$ from $P_1$; Delete $h'$ from $P_2$
	\IF{$h\neq h'$}
	\STATE Append $h'$ to $Q_1$; Append $h$ to $Q_2$;
	\ENDIF
	\ENDWHILE	

	\RETURN $(Q_1,Q_2)$.

		 \end{algorithmic}
		\end{algorithm}

		In this next example, we show how Algorithm~\ref{algo:2agent-DL-Nash} is used to compute a preference profile that is in DL-Nash equilibrium. The example also shows that it can be the case that one preference profile is in DL-Nash equilibrium and the other is not, even if both profiles yield the same outcome.

		\begin{example}[Computing a threat profile]\label{example:manip1}
			\begin{align*}
				\succ_1:\quad& h_1,h_2,h_3,h_4&\succ_2:\quad& h_2,h_3,h_1,h_4
				\end{align*}

		We now use Algorithm~\ref{algo:2agent-DL-Nash} to compute a preference profile $(\succ_1',\succ_2')$ that is in DL-Nash equilibrium:
$\succ_1'=h_1,h_2,h_3,h_4$ and 
		$\succ_2'= h_2,h_1, h_3,h_4$. Note that $PS(\succ_1',\succ_2')=\begin{pmatrix}
		  1 & 0 & 1/2 & 1/2 \\
		  0 & 1 & 1/2 & 1/2
		 \end{pmatrix}.$
			Although $PS(\succ_1,\succ_2)=PS(\succ_1',\succ_2')$, we see that $(\succ_1',\succ_2')$ is in DL-Nash equilibrium but
		$(\succ_1,\succ_2)$ is not!
		
			
		\end{example}


		%

		Next we show how our identified links with sequential allocation allocation of indivisible houses leads us to another Nash equilibrium profile called the \emph{crossout profile}. The algorithm to compute the crossout profile is stated as  Algorithm~\ref{algo:2agent-DL-crossout-Nash}.

			\begin{algorithm}[h!]
				  \caption{Crossover profile DL-Nash equilibrium for $2$ agents (which also is an EU Nash equilibrium)}
				  \label{algo:2agent-DL-crossout-Nash}
				\small
				\renewcommand{\algorithmicrequire}{\wordbox[l]{\textbf{Input}:}{\textbf{Output}:}}
				 \renewcommand{\algorithmicensure}{\wordbox[l]{\textbf{Output}:}{\textbf{Output}:}}
				\begin{algorithmic}
					\REQUIRE $(\{1,2\},H,(\succ_1,\succ_2))$
					\ENSURE The \emph{``crossout profile''} $(Q_1,Q_2)$ where $Q_i$ is the preference list of agent $i$ for $i\in \{1,2\}$.
				\end{algorithmic}
						\footnotesize
				\algsetup{linenodelimiter=\,}
				  \begin{algorithmic}[1]
			\STATE Let $P_i'$ be the order-preserving bisection of preference list of agent $i\in \{1,2\}$
			\STATE
			Initialise $Q_1'$ and $Q_2'$ to empty lists.
			\WHILE{$P_1'$ and $P_2'$ are not empty}
			\STATE Let $h= last(P_1')$; Prepend $h$ to $Q_2'$; Delete $h$ from $P_1'$ and $P_2'$;
						\STATE Let $h= last(P_2')$; Prepend $h$ to $Q_1'$; Delete $h$ from $P_1'$ and $P_2'$;
			\ENDWHILE	
			\STATE extend $Q_1'$ and $Q_2'$ to have the consecutivity property but the same allocation.
			\STATE change $Q_1$ and $Q_2'$ via  order-preserving join of $Q_1'$ and $Q_2'$.
			\RETURN $(Q_1,Q_2)$.
			
				 \end{algorithmic}
				\end{algorithm}

			%
			%

			In Algorithm~\ref{algo:2agent-DL-crossout-Nash}, the Nash equilibrium problem for PS is changed into the same problem for sequential allocation by changing each house into a half house. The idea behind the crossout profile for the sequential allocation setting is that no agent will choose the least preferred object unless it is the only object left. Thus agent $2$ will be forced to get the least preferred object of agent $1$~\citep{LeSt12a,KoCh71a}. In Algorithm~\ref{algo:2agent-DL-crossout-Nash}, we use this idea recursively to build sequences of objects 	$Q_1'$ and $Q_2'$ for each agent that are allocated to them. 
If one agent gets a half house and the other agent gets the other half house, it can be proved that the positions of the half houses in $Q_1'$ and $Q_2'$ are same. This sequence of objects for each agent are then extended to preferences that give the same allocations under sequential allocation and which also satisfy the consecutivity property. The preferences for sequential allocation are then transformed via order-preserving join to obtain the crossover Nash equilibrium profile for the PS rule. By Lemma~\ref{l:altbr2psbr}, the preference profile is in Nash equilibrium. Next we show that the threat profile and crossout profile are different and may also give different assignments.
								\begin{example}[ Crossout profile]\label{example:crossover}
							Consider the following profile.
							\begin{align*}
								\succ_1:\quad& h_1,h_2,h_3,h_4&\succ_2:\quad& h_2,h_3,h_1,h_4
								\end{align*}

						We now use Algorithm~\ref{algo:2agent-DL-crossout-Nash} to compute a preference profile $(\succ_1',\succ_2)'$ that is in DL-Nash equilibrium where 
$\succ_1'= h_2,h_1,h_4,h_3$ and $\succ_2'= h_2,h_3 ,h_4,h_1$. Note that $						PS(\succ_1',\succ_2')=\begin{pmatrix}
													  1 & 1/2 &0 & 1/2 \\
																			  0 & 1/2 & 1 & 1/2
																			 \end{pmatrix}.$ 
The crossout Nash equilibrium profile is different from the threat Nash equilibrium profile for the problem instance. 
						
						

						\end{example}

	The complexity of computing a Nash equilibrium profile for more than two agents still remains open. However we have presented a positive result for two agents --- a case which captures various fair division scenarios.

	\section{Experiments}


	In this section, we examine the \emph{likelihood} that at least one agent would have an incentive to misreport his preferences to get more expected utility. To gain insight into
	this issue we have performed a series of experiments to determine the frequency
	that, for a given number of agents and houses, a profile will have a
	beneficial strategic reporting opportunity for a single agent. \footnote{Independent from our work, \citet{Phil13a} also examined how susceptible PS can be to manipulation. \citet{HKV13a} conducted laboratory experiments which do look at the manipulability of PS mathematically but according to the strategic behaviour of humans.}


	In order to preform this experiment we need to generate preferences and utilities
	for each of the agents.
	We consider two different models to generate profiles.
	(i) In the \emph{Impartial Culture (IC)} model, the assumption is that for each agent and a given number of
	houses, each of the $|H|!$ preference orders over the houses is equally likely ($\frac{1}{|H|!}$). (ii) In the \emph{Uniform Single Peaked (USP)}, the assumption is that
	all single peaked preference profiles are equally likely.  Single peaked preferences
	are a profile restriction introduced by \citet{Blac48a} and well studied in the
	social choice literature.  Informally, in a single peaked profile, given all possible 3-sets of houses, no agent ever ranks some particular house last in all 3 sets that it appears.
	%

	%
	%
	%

	In order to evaluate if an agent has a better response  we need to assign utilities
	to the individual houses for each agent.  While there are a number of ways to model
	utility we have selected the following mild restrictions on utilities in order to gain an understanding of the manipulation opportunities.
	(i) In the \emph{Random} model, we uniformly at random generate a real number between $0$ and $1$ for each
	house that is compatible with the generated preference order.
	We normalize these utilities such that each agent's utility sums to a constant value that is the same
	for all agents.  In our experiments each agent's utility sums to the number of houses in the instance.
	(ii) In the \emph{Borda} model, we assign $|H|-1$ utility to the first house,
	$|H|-2$ to the second house, down to $0$ utility for the least preferred house.
	(iii) In the \emph{Exponential (Exp)} model, we assign utility
	$2^{|H|-1}$ to the first house, $2^{|H|-2}$ to the second house, down to
	$0$ utility for the least preferred house.

	We generated for each pair in $|N| = \{1, \dots, 8\} \times |H| = \{1, \dots, 8\}$ 1,000 profiles according to a utility
	and preference distribution.
	For each of these instances, we searched to see if any agent could get more utility by misreporting his preferences, if so, 
	then we say that profile admitted a manipulation.  Figure \ref{fig:rand-borda} show the
	percentage of instances that were manipulable for each of the domain, utility, number of agent, and number of house combinations (Borda is omitted for space).



	Looking at Figure \ref{fig:rand-borda}, we observe that as the utility and preference models become
	more restrictive, the opportunities for a single agent to manipulate becomes smaller.  The Random-IC experiment
	yields the most frequently manipulable profiles, strictly dominating all the other runs of the experiment for every
	combination except one (Random-USP with 3 houses and 3 agents).  Each experiment with single
	peaked preferences (save one) is dominated by the experiment with the unrestricted preference
	profiles for the same utility model.
	
	\begin{figure}[ht]
	\begin{minipage}[b]{\textwidth}
	\centering
	\includegraphics[width=.3\linewidth]{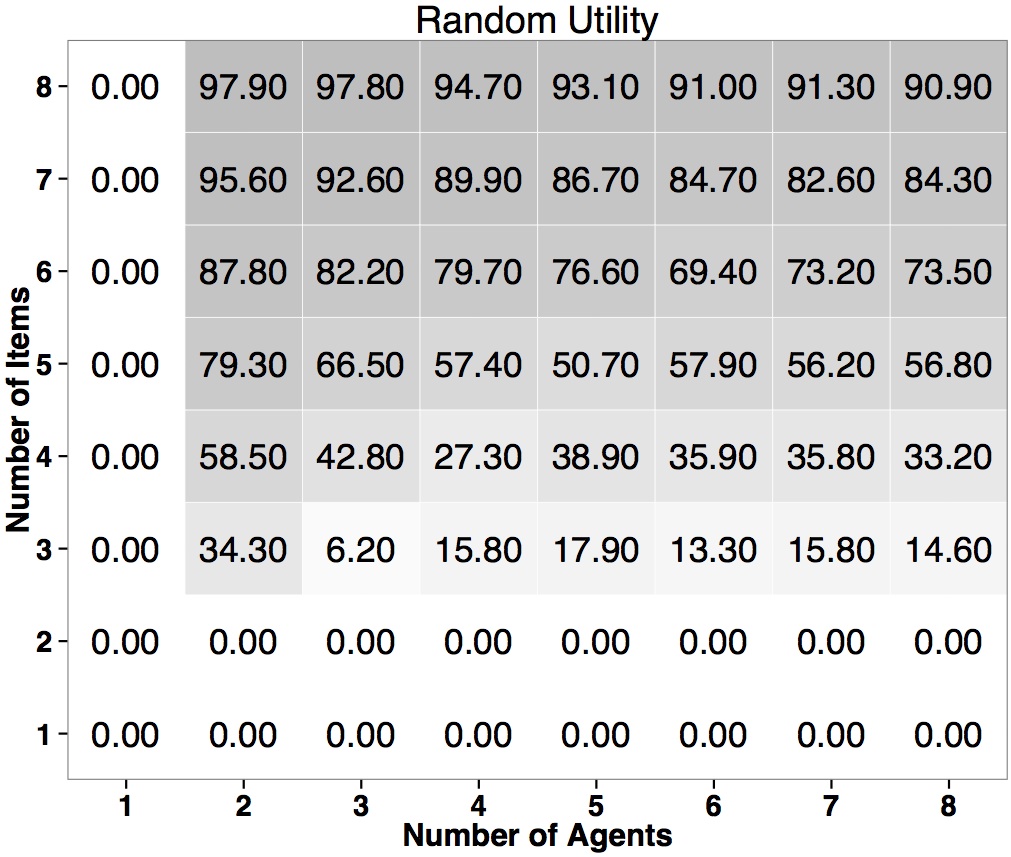}
	\hspace{0.1cm}
	\includegraphics[width=.3\linewidth]{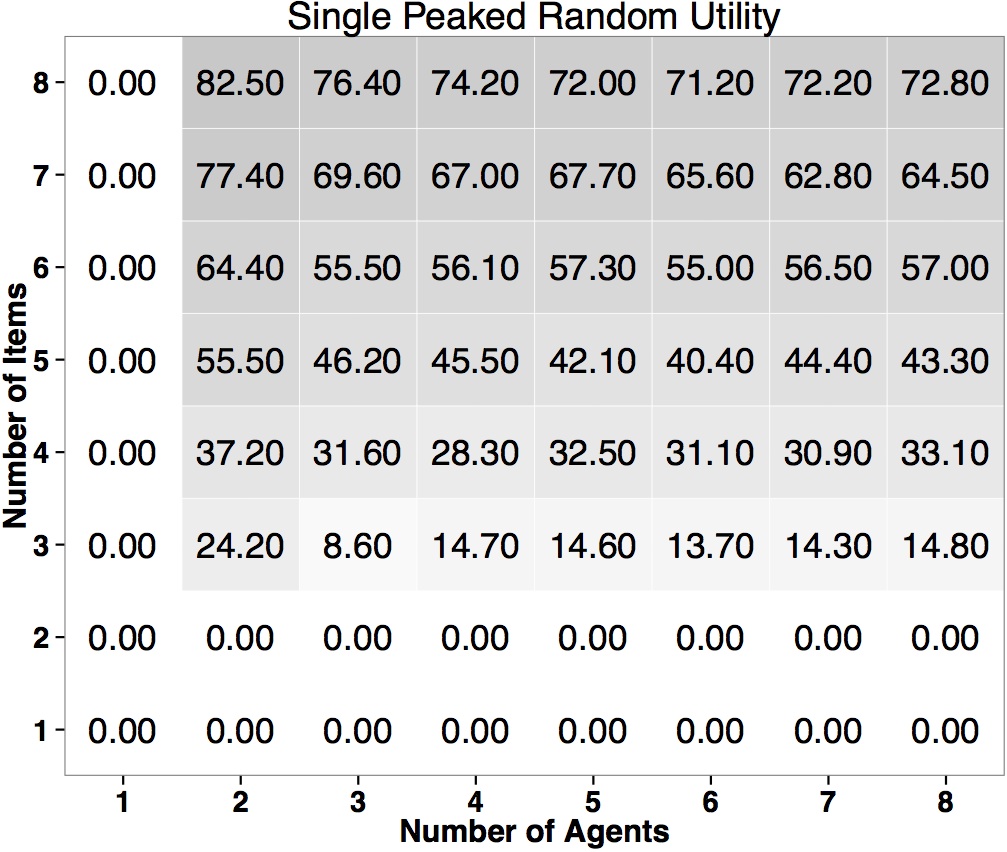}
	\end{minipage}
	\vspace{0.1cm}
	\begin{minipage}[b]{\textwidth}
	\centering
	 \includegraphics[width=.3\linewidth]{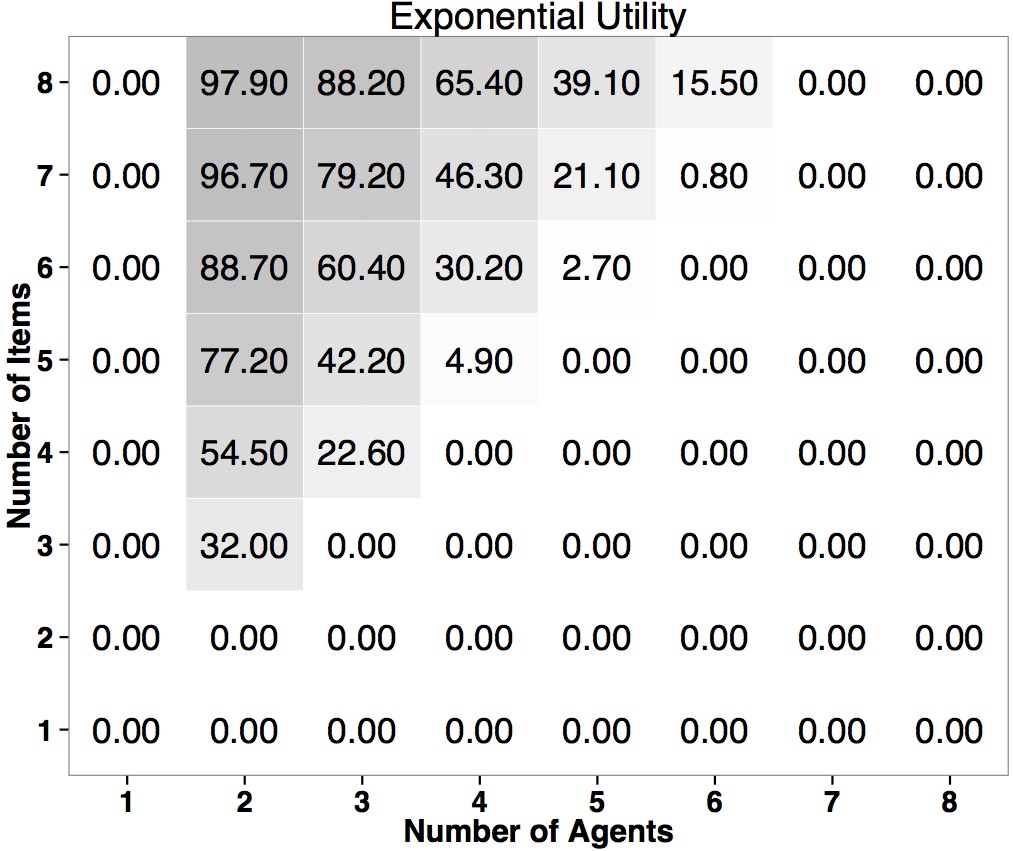}	
	 \hspace{0.1cm}
	\includegraphics[width=.3\linewidth]{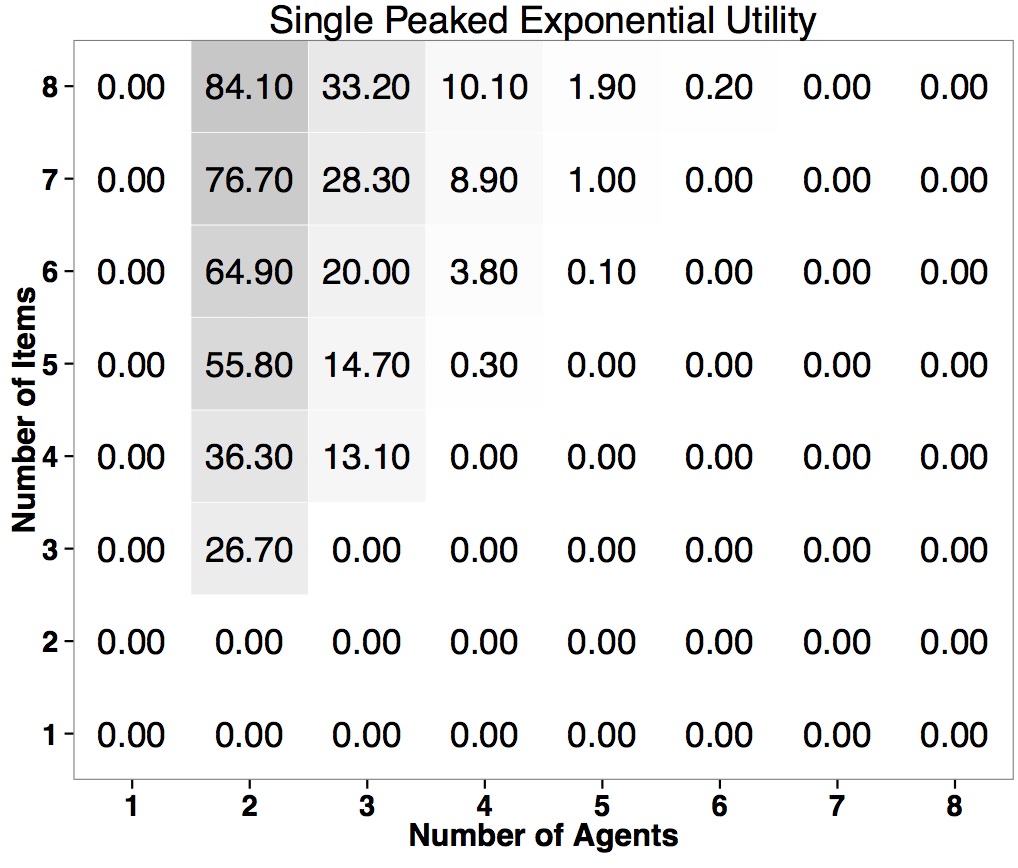}
	\end{minipage}

	 \caption{Heatmap showing the percentage of manipulable instances for Random Utility
	 and the Single Peaked Random Utility models on the left and Exponential Utility
	 and the Single Peaked Exponential Utility models on the right.}
	 \label{fig:rand-borda}
	\end{figure}

	The PS rule is strategyproof with respect to the DL relation in the case where the number of agents and the number of houses
	are equal.  Our experiment with the Exp model (which is similar to the DL relation but not exact) found
	no manipulable instances when the number of agents is less than or equal to the number of houses.
	It is encouraging that the manipulation opportunities for Exp-USP are so low.  In this setting each agent is
	valuing the houses along the same axis of preference and prefers their first choice exponentially more than
	their second choice.
	As the number of
	houses relative to the number of agents grows, the opportunities to manipulate increase, maximizing
	around $99\%$.

	\section{Conclusions}

	We conducted a detailed computational analysis of strategic aspects of the PS rule.
	Our study leads to a number of new research directions. PS is well-defined even for indifferences~\citep{KaSe06a}. It will be interesting to extend our results for strict preferences to the case with ties. Two interesting problems are still open. Firstly, What is the complexity of computing an expected utility best response for more than two agents? The problem is particularly intriguing because even for the related and conceptually simpler setting of discrete allocation, computing an expected utility  best response for more than two agents has remained an open problem~\citep{BoLa11a}.
Another problem is the complexity of computing a Nash equilibrium for more than two agents. 	
It will also be interesting to examine coalitional manipulations and coalitional Nash equilibria.
	Finally, an analysis of Nash dynamics under the PS rule is an intriguing research problem.

	\renewcommand*{\bibfont}{\small}


\begin{thebibliography}{}

	\bibitem[\protect\citeauthoryear{Ackermann, Goldberg, Mirrokni, R{\"o}glin, and
	  V{\"o}cking}{Ackermann et~al\mbox{.}}{2011}]{AGM+11a}
	{\sc Ackermann, H.}, {\sc Goldberg, P.~W.}, {\sc Mirrokni, V.~S.}, {\sc
	  R{\"o}glin, H.}, {\sc and} {\sc V{\"o}cking, B.} 2011.
	\newblock Uncoordinated two-sided matching markets.
	\newblock {\em SIAM Journal on Computing\/}~{\em 40,\/}~1, 92--106.

	\bibitem[\protect\citeauthoryear{Aziz, Brandt, and Brill}{Aziz
	  et~al\mbox{.}}{2013a}]{ABB13b}
	{\sc Aziz, H.}, {\sc Brandt, F.}, {\sc and} {\sc Brill, M.} 2013a.
	\newblock The computational complexity of random serial dictatorship.
	\newblock {\em Economics Letters\/}~{\em 121,\/}~3, 341--345.

	\bibitem[\protect\citeauthoryear{Aziz, Brandt, and Stursberg}{Aziz
	  et~al\mbox{.}}{2013b}]{ABS13a}
	{\sc Aziz, H.}, {\sc Brandt, F.}, {\sc and} {\sc Stursberg, P.} 2013b.
	\newblock On popular random assignments.
	\newblock In {\em Proceedings of the 6th International Symposium on Algorithmic
	  Game Theory (SAGT)}, {B.~V\"ocking}, Ed. Lecture Notes in Computer Science
	  (LNCS) Series, vol. 8146. Springer-Verlag, 183--194.

	\bibitem[\protect\citeauthoryear{Black}{Black}{1948}]{Blac48a}
	{\sc Black, D.} 1948.
	\newblock On the rationale of group decision-making.
	\newblock {\em Journal of Political Economy\/}~{\em 56,\/}~1, 23--34.

	\bibitem[\protect\citeauthoryear{Bogomolnaia and Heo}{Bogomolnaia and
	  Heo}{2012}]{BoHe12a}
	{\sc Bogomolnaia, A.} {\sc and} {\sc Heo, E.~J.} 2012.
	\newblock Probabilistic assignment of objects: Characterizing the serial rule.
	\newblock {\em Journal of Economic Theory\/}~{\em 147}, 2072--2082.

	\bibitem[\protect\citeauthoryear{Bogomolnaia and Moulin}{Bogomolnaia and
	  Moulin}{2001}]{BoMo01a}
	{\sc Bogomolnaia, A.} {\sc and} {\sc Moulin, H.} 2001.
	\newblock A new solution to the random assignment problem.
	\newblock {\em Journal of Economic Theory\/}~{\em 100,\/}~2, 295--328.

	\bibitem[\protect\citeauthoryear{Bouveret and Lang}{Bouveret and
	  Lang}{2011}]{BoLa11a}
	{\sc Bouveret, S.} {\sc and} {\sc Lang, J.} 2011.
	\newblock A general elicitation-free protocol for allocating indivisible goods.
	\newblock In {\em Proceedings of the 22 International Joint Conference on
	  Artificial Intelligence (IJCAI)}. 73--78.

	\bibitem[\protect\citeauthoryear{Budish, Che, Kojima, and Milgrom}{Budish
	  et~al\mbox{.}}{2013}]{BCKM12a}
	{\sc Budish, E.}, {\sc Che, Y.-K.}, {\sc Kojima, F.}, {\sc and} {\sc Milgrom,
	  P.} 2013.
	\newblock Designing random allocation mechanisms: {T}heory and applications.
	\newblock {\em American Economic Review\/}.
	\newblock Forthcoming.

	\bibitem[\protect\citeauthoryear{Cho}{Cho}{2012}]{Cho12a}
	{\sc Cho, W.~J.} 2012.
	\newblock Probabilistic assignment: A two-fold axiomatic approach.
	\newblock Unpublished manuscript.

	\bibitem[\protect\citeauthoryear{Ekici and Kesten}{Ekici and
	  Kesten}{2012}]{EkKe12a}
	{\sc Ekici, O.} {\sc and} {\sc Kesten, O.} 2012.
	\newblock An equilibrium analysis of the probabilistic serial mechanism.
	\newblock Tech. rep., {\"O}zye{\u g}in University, Istanbul. May.

	\bibitem[\protect\citeauthoryear{Faliszewski, Hemaspaandra, and
	  Hemaspaandra}{Faliszewski et~al\mbox{.}}{2010}]{FHH10a}
	{\sc Faliszewski, P.}, {\sc Hemaspaandra, E.}, {\sc and} {\sc Hemaspaandra, L.}
	  2010.
	\newblock Using complexity to protect elections.
	\newblock {\em Communications of the ACM\/}~{\em 53,\/}~11, 74--82.

	\bibitem[\protect\citeauthoryear{Faliszewski and Procaccia}{Faliszewski and
	  Procaccia}{2010}]{FaPr10a}
	{\sc Faliszewski, P.} {\sc and} {\sc Procaccia, A.~D.} 2010.
	\newblock {AI}'s war on manipulation: {A}re we winning?
	\newblock {\em AI Magazine\/}~{\em 31,\/}~4, 53--64.

	\bibitem[\protect\citeauthoryear{G{\"a}rdenfors}{G{\"a}rdenfors}{1973}]{Gard73b}
	{\sc G{\"a}rdenfors, P.} 1973.
	\newblock Assignment problem based on ordinal preferences.
	\newblock {\em Management Science\/}~{\em 20}, 331--340.

	\bibitem[\protect\citeauthoryear{Hugh-Jone, Kurino, and Vanberg}{Hugh-Jone
	  et~al\mbox{.}}{2013}]{HKV13a}
	{\sc Hugh-Jone, D.}, {\sc Kurino, M.}, {\sc and} {\sc Vanberg, C.} 2013.
	\newblock An experimental study on the incentives of the probabilistic serial
	  mechanism.
	\newblock Tech. Rep. SP II 2013--204, Social Science Research Center Berlin
	  (WZB). May.

	\bibitem[\protect\citeauthoryear{Hylland and Zeckhauser}{Hylland and
	  Zeckhauser}{1979}]{HyZe79a}
	{\sc Hylland, A.} {\sc and} {\sc Zeckhauser, R.} 1979.
	\newblock The efficient allocation of individuals to positions.
	\newblock {\em The Journal of Political Economy\/}~{\em 87,\/}~2, 293--314.

	\bibitem[\protect\citeauthoryear{Katta and Sethuraman}{Katta and
	  Sethuraman}{2006}]{KaSe06a}
	{\sc Katta, A.-K.} {\sc and} {\sc Sethuraman, J.} 2006.
	\newblock A solution to the random assignment problem on the full preference
	  domain.
	\newblock {\em Journal of Economic Theory\/}~{\em 131,\/}~1, 231--250.

	\bibitem[\protect\citeauthoryear{Kohler and Chandrasekaran}{Kohler and
	  Chandrasekaran}{1971}]{KoCh71a}
	{\sc Kohler, D.~A.} {\sc and} {\sc Chandrasekaran, R.} 1971.
	\newblock A class of sequential games.
	\newblock {\em Operations Research\/}~{\em 19,\/}~2, 270--277.

	\bibitem[\protect\citeauthoryear{Kojima}{Kojima}{2009}]{Koji09a}
	{\sc Kojima, F.} 2009.
	\newblock Random assignment of multiple indivisible objects.
	\newblock {\em Mathematical Social Sciences\/}~{\em 57,\/}~1, 134---142.

	\bibitem[\protect\citeauthoryear{Levine and Stange}{Levine and
	  Stange}{2012}]{LeSt12a}
	{\sc Levine, L.} {\sc and} {\sc Stange, K.~E.} 2012.
	\newblock How to make the most of a shared meal: Plan the last bite first.
	\newblock {\em The American Mathematical Monthly\/}~{\em 119,\/}~7, 550--565.

	\bibitem[\protect\citeauthoryear{Philipp}{Philipp}{2013}]{Phil13a}
	{\sc Philipp, B.} 2013.
	\newblock Simulation of boundedly rational manipulation strategies in one-sided
	  matching markets.
	\newblock M.S.\ thesis, Faculty of Economics, University of Zurich.

	\bibitem[\protect\citeauthoryear{Saban and Sethuraman}{Saban and
	  Sethuraman}{2013a}]{SeSa13a}
	{\sc Saban, D.} {\sc and} {\sc Sethuraman, J.} 2013a.
	\newblock House allocation with indifferences: a generalization and a unified
	  view.
	\newblock In {\em Proceedings of the 14th ACM Conference on Electronic Commerce
	  (ACM-EC)}. 803--820.

	\bibitem[\protect\citeauthoryear{Saban and Sethuraman}{Saban and
	  Sethuraman}{2013b}]{SaSe13b}
	{\sc Saban, D.} {\sc and} {\sc Sethuraman, J.} 2013b.
	\newblock A note on object allocation under lexicographic preferences.
	\newblock {\em Journal of Mathematical Economics\/}.

	\bibitem[\protect\citeauthoryear{Schulman and Vazirani}{Schulman and
	  Vazirani}{2012}]{ScVa12a}
	{\sc Schulman, L.~J.} {\sc and} {\sc Vazirani, V.~V.} 2012.
	\newblock Allocation of divisible goods under lexicographic preferences.
	\newblock Tech. Rep. arXiv:1206.4366, arXiv.org.

	\bibitem[\protect\citeauthoryear{Yilmaz}{Yilmaz}{2010}]{Yilm10a}
	{\sc Yilmaz, O.} 2010.
	\newblock The probabilistic serial mechanism with private endowments.
	\newblock {\em Games and Economic Behavior\/}~{\em 69,\/}~2, 475--491.

	\end{thebibliography}



	\newpage
	\appendix

	\section{Pseudocode of PS}

	We write the formal definition of PS from \citep{Koji09a} as an algorithm. For any $h\in H'\subset H$, let $N(h,H')=\{i\in N\midd a\succ_i b \text{  for every } b\in H'\}$ be the set of agents whose most preferred house in $H'$ is $h$. PS is defined as Algorithm~\ref{algo:subroutine}.

	\begin{algorithm}[htb]
	  \caption{PS}
	  \label{PS}
	\renewcommand{\algorithmicrequire}{\wordbox[l]{\textbf{Input}:}{\textbf{Output}:}}
	 \renewcommand{\algorithmicensure}{\wordbox[l]{\textbf{Output}:}{\textbf{Output}:}}
	\begin{algorithmic}
		\REQUIRE $(N,H,\succ)$
		\ENSURE $p$ the random assignment returned by PS
	\end{algorithmic}
	\algsetup{linenodelimiter=\,}
	  \begin{algorithmic}[1]

	\STATE $s\longleftarrow 0$ ($s$ is the stage of the algorithm)
	\STATE $H^O\longleftarrow H$; $t^0\longleftarrow 0$; $p_{ih}^0\longleftarrow 0$ for all $i\in N$ and $h\in H$.

	\WHILE{$H^s\neq \emptyset$}
	\STATE $t^{s+1}(h)=\sup\{t\in [0,|H|]\midd \sum_{i\in N} p_{ih}^s+|N(h,H^s)|(t-t^s)<1\}$
	\STATE $t^{s+1}\longleftarrow \min_{h\in H^s}t^s(h)$
	\STATE $H^s=H^s\setminus \{h\in H^{s-1}\midd t(h)=t^s\}$
	\FOR{all $i\in N$ and $h\in H$}
	\IF{$ i\in N(h,H^s)$}
	\STATE  $p_{ih}^{s+1} \longleftarrow  p_{ih}^{s}+t^{s+1}-t(s) $
	\ELSE
	\STATE $ p_{ih}^{s+1} \longleftarrow  p_{ih}^{s} $
	\ENDIF
	\ENDFOR

	\STATE $s\longleftarrow s+1$
	\ENDWHILE

	\RETURN $p=p^s$

	 \end{algorithmic}
	\label{algo:subroutine}
	\end{algorithm}

		\section{Expected utility best response for the alternation policy}\label{s:alt_bestresp}
		In this section we recall the best response algorithm
		proposed in \citep{KoCh71a} as we will use it to derive
		the best response algorithm for the $\ps$ algorithm.

		We denote the algorithm from \citep{KoCh71a}  $\BestEUResAlgo$.
		%
		%
		In particular, we describe $\BestEUResAlgo$ for the special case $k_i=1$ and $n_i=2$ so that we
		follow the alternation policy. We also assume that the number of houses is even
		as this is sufficient for our purposes.  These restrictions simplify the
		algorithm.


		Following Kohler and Chandrasekaran~\citep{KoCh71a}, we use a  matrix $V = V_{i,j}$, $i=1,2$, $j=1,\ldots,m$, where $V_{i,j}$
	represents the utility value that the $i$-th
	player will gain if he selects the $h_j$ object.
	In our case, we  assume that $V_{i,j} = u_{i,j}$, $i=1,2$, $j=1,\ldots,m$,
	such that $u_{i,j} \in \mathbb{R}$ and $u_{i,j} > u_{i,j'}$ iff $h_j \pref_i h_{j'}$.
	As $\pref_2$ ranks houses $h_j$, $j=1,\ldots,m$, lexicographically, we have $V_{2,j} \geq V_{2,j+1}$, $j=1,\ldots,m-1$.
	Algorithm~\ref{algo:BestEUResAlgo} shows a pseudocode for the simplified version of  $\BestEUResAlgo$.

	 	\begin{algorithm}[htb]
	 	  \caption{$\BestEUResAlgo$ for sequential allocation for two agents}
	 	  \label{algo:BestEUResAlgo}
	 	\small
	 	\renewcommand{\algorithmicrequire}{\wordbox[l]{\textbf{Input}:}{\textbf{Output}:}}
	 	 \renewcommand{\algorithmicensure}{\wordbox[l]{\textbf{Output}:}{\textbf{Output}:}}
	 	\begin{algorithmic}
	 		\REQUIRE $(V_{i,j}, i=1,2, j=1,\ldots,m)$
	 		\ENSURE the set of houses allocated to agent $1$ as a result of his best response.
	 	\end{algorithmic}
	 	\algsetup{linenodelimiter=\,}
	 	  \begin{algorithmic}[1]
	\small
	\FOR{$k\in [1,m/2]$}
	 \STATE $ I^k \longleftarrow \{h_{2k-1},h_{2k}\}$
	\ENDFOR
	\IF {$\prefinvfn{1}{h_1} < \prefinvfn{1}{h_2}$ }
	\STATE ${J}^1 \longleftarrow \{h_1\}$
	\ELSE
	\STATE ${J}^1 \longleftarrow \{h_2\}$
	\ENDIF
	\FOR{$k\in [2,m/2]$}
	 \STATE ${J}^k \longleftarrow \{h_{2k-1},h_{2k}\}$
	\ENDFOR 	
	 	\STATE $\overline{J}^1\longleftarrow J^1$
	 \FOR{$k\in [2,m/2]$}
	\STATE $\overline{J}^k \longleftarrow \{h_{j}| h_{j} \in \overline{J}^{k-1} \cup J^k; V_{1j} \geq k \mathrm{th\ maximal\ of\ } \{V_{1l}| h_{l}  \in \overline{J}^{k-1} \cup J^k\}\}$
	\ENDFOR
	 \RETURN $\overline{J}^{m/2}$.
	 	 \end{algorithmic}
	 	\end{algorithm}

	%
	%
		We refer to $\overline{J}^k$  as an ordered set formed at the $k$th stage of $\BestEUResAlgo$.
		The ordered set $\overline{J}^{m/2}$ is the optimal set of houses for agent $1$ to choose, and agent $1$ must choose
		them in the lexicographic  order.
		We denote $\BestEURes = \BestEUResAlgo(V)$. Note that the number of houses in $\BestEURes$ is $m/2$.


		\begin{example}~\label{exm:kohler_run}
		Consider two agents with preferences $\pref_1 = h_5, h_6, h_1, h_3, h_4, h_2$
		and  $\pref_2 = h_1, h_2, h_3, h_4, h_5, h_6$.
	First, we form a matrix $V$. We select arbitrary numbers $u_{i,j}$ that satisfy conditions above, e.g.

	\[
	V =\begin{pmatrix}
		  4&1&3&2&6&5\\
		   6&5&4&3&2&1
	 \end{pmatrix}
	\]

		The following table shows an execution of the algorithm
		on this example over profiles $\pref_1$ and $\pref_2$.
	\begin{center}
			\scalebox{0.8}{
		$\begin{array}{|c|c|c|}
		\hline
		 I^1 & I^2& I^3  \\
		\hline
		\hline
		   \{h_1, h_2\} & \{h_3, h_4\} & \{h_5, h_6\}\\
		\hline
		\hline
		 J^1 & J^2& J^3  \\
		\hline
		\hline
		   \{h_1 \} & \{h_3, h_4\} & \{h_5, h_6\}\\
		\hline
		\hline
		 \overline{J}^1 & \overline{J}^2& \overline{J}^3  \\
		\hline
		\hline
		  \{h_1 \}  & \{h_1, h_3 \} &  \{h_1, h_5, h_6 \}  \\
		   \hline
		\end{array}
		 $
		}
	\captionof{table}{An execution of $\BestEUResAlgo$ on Example~\ref{exm:kohler_run}.}
	\end{center}

		$\BestEURes = \overline{J}^{3} =  \{h_1, h_5, h_6 \}$.
		\end{example}

		Given $\BestEURes$ we define a profile that corresponds to the best response
		$\prefbest_1$.
		By $\BestEURes(i)$ we refer to the house at the $i$th position.
	 First, we rank houses in $\BestEURes$ in the same order as they occur in $\BestEURes$, so that
	$\prefbest_1 = (\BestEURes(1),\ldots, \BestEURes(m/2)$.
	Then, after $\BestEURes(m/2)$, we rank houses that agent 2 gets in the same order as agent 2 obtains them. 
		In Example~\ref{exm:kohler_run}, $\prefbest_1 = h_1, h_5 ,h_6, h_2, h_3, h_4$.

	\section{A best response without the  consecutivity property (example)}
		Next we provide an example that shows
	that a best response returned by Algorithm~\ref{algo:BestEUResAlgo} over  $\prefcloned_1$ and $\prefcloned_2$ might not have the  consecutivity property.
		\begin{example}\label{exm:transeubr_alt2ps}
		Consider two agents from Example~\ref{exm:eqalt2ps}.
		We recall that if we split all houses into halves
		then we obtain profiles:
		$\prefcloned_1 = h_5^1,h_5^2,h_6^1,h_6^2,h_1^1,  h_1^2 ,h_3^1,h_3^2,h_4^1,h_4^2,h_2^1,h_2^2$ and
		$\prefcloned_2 = h_1^1  ,  h_1^2 , h_2^1 , h_2^2 , h_3^1 , h_3^2  , h_4^1  , h_4^2  , h_5^1  , h_5^2  , h_6^1 , h_6^2.$

	 A matrix $V$ is the following
	\setcounter{MaxMatrixCols}{20}
	\[
	V =\begin{pmatrix}
		  4&4&1 &1&3 &3&2 &2&6 &6&5 &5\\
		   6&6&5 &5&4 &4&3 &3&2 &2&1 &1
	 \end{pmatrix}
	\]

		Table~\ref{exm:noncons} shows an execution of $\BestEUResAlgo$
		over profiles $\prefcloned_1$ and $\prefcloned_2$.
		$\BestEURes = \overline{J}^{6} =  \{h_1^1, h_3^1, h_5^1, h_5^2, h_6^1, h_6^2\}$.
		We extend $\BestEURes$ with houses that are not allocated to agent $1$ and obtain

		\begin{align*}
		 {\prefclonedbest_1} =   h_1^1 , h_3^1 , h_5^1 , h_5^2 , h_6^1 , h_6^2 , h_1^2 , h_2^1 , h_2^2 ,  h_3^2 , h_4^1 , h_4^2.
		\end{align*}
		\begingroup
	\everymath{\scriptstyle}
	\begin{center}
	\large
			\scalebox{0.7}{	
	$\begin{array}{|c|c|c|c|c|c|}

		\hline
		 I^1 & I^2& I^3 & I^4 & I^5& I^6 \\
		\hline
		\hline
		   \{h_1^1, h_1^2\} & \{h_2^1, h_2^2\} & \{h_3^1, h_3^2\}&  \{h_4^1, h_4^2\} & \{h_5^1, h_5^2\} & \{h_6^1, h_6^2\}\\
		\hline
		\multicolumn{6}{c}{} \\
		\hline
		 J^1 & J^2& J^3  & J^4 & J^5& J^6\\
		\hline
		\hline
		   \{h_1^1 \} & \{h_2^1, h_2^2\} & \{h_3^1, h_3^2\}&  \{h_4^1, h_4^2\} & \{h_5^1, h_5^2\} & \{h_6^1, h_6^2\}\\
		\hline
		\multicolumn{6}{c}{} \\
		\hline
		 \overline{J}^1 & \overline{J}^2& \overline{J}^3  &  \overline{J}^4 & \overline{J}^5& \overline{J}^6  \\
		\hline
		\hline
		  \{h_1^1 \}  & \{h_1^1, h_2^1 \} &  \{h_1^1, h_3^1, h_3^2 \}   &  \{h_1^1, h_3^1, h_3^2, h_4^1 \}  &  \{h_1^1, h_3^1, h_3^2, h_5^1, h_5^2 \}  &  \{h_1^1, h_3^1, h_5^1, h_5^2, h_6^1, h_6^2 \}\\
		   \hline
		\end{array}
		 $
		}
	\captionof{table}{An execution of $\BestEUResAlgo$ over profiles $\prefcloned_1$ and $\prefcloned_2$.\label{exm:noncons}}
	\end{center}
	\endgroup
	Unfortunately,  ${\prefclonedbest_1}$  does not have the consecutivity property and Lemma~\ref{l:altbr2psbr} can not be applied.
	Note that agent 2 gets $ \{h_1^2, h_2^1, h_2^2, h_3^2, h_4^1, h_4^2\}$.
		\end{example}

			In the next section, we show that we can always find another ${\prefclonedbest_1}$ that has the consecutivity property.

		\section{Expected utility best response for the PS mechanism (full proof).}


		In this section, we demonstrate that given $\prefcloned_1$ and $\prefcloned_2$
		we can always find the expected utility best response to $\prefcloned_2$ that has
		the consecutivity property. To do so, we first run $\BestEUResAlgo$
		to obtain $\BestEURes$. Then we demonstrate that it can be
		modified and extended to a profile over $\hclone$ that has the
		consecutivity property.

	 Given  $\BestEURes$, we denote 	the ordered set of houses allocated to agent $j$ $\bestalloc_{j}$, $j=1,2$.
	Note that $\bestalloc_{1} = \BestEURes$.
	Then $\bestalloc_1(i)$ and $\bestalloc_2(i)$
		are houses that are allocated to agent $1$ and agent $2$, respectively,
		in the $i$th round of the alternation policy.
		\begin{example}
		Consider Example~\ref{exm:transeubr_alt2ps}.
	\[ \bestalloc_1= \{h_1^1, h_3^1, h_5^1, h_5^2, h_6^1, h_6^2\}.\] and \[\bestalloc_2= \{h_1^2, h_2^1, h_2^2, h_3^2, h_4^1, h_4^2\}.\]
		\end{example}
	 We say that
	$\bestalloc$ has the  consecutivity property iff for all $h_i^1,h_i^2 \in \bestalloc$, $h_i^1$ and $h_i^2$
	are ordered consecutively.
		We say that a half-house of $h_i$ is allocated to agent $1$ if and only if agent $1$
		gets $h_i^1$ and agent $2$ gets $h_i^2$.
		We say that a full-house $h_i$ is allocated to agent $1$ if and only if agent $1$
		gets $h_i^1$ and $h_i^2$.


	In the proof we often consider an ordered set of houses $\{h_{i_1}^j, \ldots, h_{i_p}^j\}$
	that obeys the following property: $\{h_{i_1}^j \succ_2 \ldots \succ_2 h_{i_p}^j\}$.
	We will say these houses are lexicographically ordered as agent 2 orders his houses
	w.r.t. the lexicographic order by our assumption in Section~\ref{s:connection_alt_ps}.

	First, we give an overview of the construction. Our construction is motivated by an
	observation that if $\bestalloc_1$ and $\bestalloc_2$ have the consecutivity property
	and half houses are obtained by agent 1 and 2 at the same round then it is
	straightforward to extend  $\bestalloc_1$ to ${\prefclonedbest_1}$ over $\hclone$ that has the
	consecutivity property. Consider the following example.

	\begin{example}
	Suppose  $ \bestalloc_{1} = \{h_2^1, h_2^2, h_3^1, h_4^1, h_4^2, h_6^1\}$
	and  $ \bestalloc_{2} = \{h_1^1, h_1^2, h_3^2, h_5^1, h_5^2, h_6^2\}$.
	Note that half-houses are allocated
	in the same rounds in $\bestalloc_1$ and $\bestalloc_2$.
	The $1$st agent expected utility best response profile is
	$ {\prefclonedbest_1} = h_2^1 , h_2^2 , h_3^1  , h_4^1  , h_4^2 , h_6^1  { ,  h_1^1  , h_1^2 , {h_3^2}, h_5^1  ,  h_5^2 , { h_6^2 }  } $.
	Note that $ {\prefclonedbest_1}$ does not have consecutivity property.

	Next we demonstrate how to change $ {\prefclonedbest_1}$ so that it has the consecutivity property and leads to the same allocation.
	For each half house $h_i^1$ allocated to agent $1$ we rank $h_i^2$ right after $h_i^1$. We keep houses that are not allocated
		to agent $1$ in the end of the profile.
	In this example, we rank $h_3^2$ and $h_6^2$ after $h_3^1$ and $h_6^1$, respectively.
	We obtain the following profile:
	 ${\prefclonedbest_1} = h_2^1 , h_2^2 , h_3^1 , \mathbf{h_3^2} , h_4^1  , h_4^2 , h_6^1  , \mathbf{ h_6^2 } { ,  h_1^1  , h_1^2 , h_5^1  ,  h_5^2 } $.
		Note that inserting $h_i^2$  after $h_i^1$  does not change the allocation
		as we know that $h_i^2$ is allocated to agent $2$  at the same round as $h_i^1$
		is allocated to agent $1$. Hence, $h_i^2$ will never be the top element
		for agent $1$ at any round and  ${\prefclonedbest_1}'$ gives the same allocation as
		$\bestalloc_1$.
		\end{example}

	Based on this observation, the goal of the construction is to transform $\bestalloc_{1}$ is such a way that half houses
	of $h_i$ that are allocated to different agents are allocated to them in the same round
	while preserving allocations of both agents.
	To do so, we prove that an allocation of half houses and full houses in an execution of $\BestEUResAlgo$ follows simple patterns.
	The first property concerns full houses :  halves of full houses allocated to an agent are always allocated in consecutive rounds.
	The second key property concerns half houses.  Let $H_{G_j} = \{h_{i_1}^j,\ldots, h_{i_p}^j\}$ be the lexicographically  ordered set
	 of half houses allocated to agent $j$, $j=1,2$.
	  Then allocation of half houses obeys the following order:
	agent 1 gets $h_{i_1}^1$ at round $k_{t_1}$ then, possibly in later round $k_{t_1}'$, agent 2 gets
	$h_{i_1}^2$. Next, agent 1 gets $h_{i_2}^1$ at round $k_{t_2}$ then, possibly in later round $k_{t_2}'$, agent 2 gets
	$h_{i_1}^2$, and so on.  In other words, half houses are allocated  to agents in lexicographic order and each half houses
	is allocated to both agents before the next half houses is allocated.
	Based on these properties, we will prove that we can delay an allocation of $h_{i_h}^1$ to agent 1 till round $k_{t_h}'$ and preserve the allocations.

		We need to prove several useful properties
		of $\BestEURes$.
	%

		The next proposition states that if only half of $h_i$
		is allocated to agent $1(2)$ then this half is $h_i^1(h_i^2)$.
		We use this observation to simplify notations.

		\begin{proposition}\label{prop:halvesorder}
		If $h_i^j$ is allocated to agent $1$
		and  $h_i^{j\%2 +1}$ is allocated to agent $2$
		then $h_i^1$ is allocated to agent $1$ and $h_i^2$ is allocated to agent $2$.
		\end{proposition}
		\begin{proof}
		Follows from  $\BestEUResAlgo$ and $\prefcloned_1$ as between $h_i^1$ and $h_i^2$
		 agent $1$ always prefers $h_i^1$.
		\end{proof}

		The next lemma shows that for all full houses allocated to $i$,
		both halves are allocated in consecutive rounds.
		\begin{proposition}\label{prop:consec}
		If a full-house of $h_i$ is allocated to $1(2)$ then $h_i^1$ and $h_i^2$ are allocated to $1(2)$
		in two consecutive rounds.
		\end{proposition}
		\begin{proof}
		For agent $1$ it follows from  construction of $\bestalloc_{1}$ in   $\BestEUResAlgo$.
		For agent $2$ it follows from the definition $\prefcloned_2$, as  $h_i^1$ and $h_i^2$ are ordered
	consecutively in $\prefcloned_2$, and the fact that we use the
		alternation policy to obtain $\bestalloc_2$.
		\end{proof}
	%

		We denote $H_{G_j} = \{h_{i_1}^j,\ldots, h_{i_p}^j\}$ the lexicographically ordered set
		of half-houses allocated to agent $j$, $j=1,2$.
		We show that utilities of these houses decrease monotonically given this order.

		\begin{proposition}\label{prop:utility_order}
		$V_{1h_{i_1}^1} > \ldots > V_{1h_{i_p}^1}$ for $h_{i_1}^1,\ldots, h_{i_p}^1 \in H_{G_1}$.
		\end{proposition}
		\begin{proof}
		By contradiction, suppose that $h_{t'}^1$ is a half-house that violates
		the statement: $V_{1h_{t'-1}^1} <  V_{1h_{t'}^1}$. The equality is not possible
		as we have strict preferences over houses.
		We denote $t=t'-1$ to simplify notations.
		From $\BestEUResAlgo$, it follows that $h_{t}^1$ was added to $\overline{J}_t$
		from  ${J}_t = \{h_{t}^1, h_{t}^2\}$ at the stage $t$
		and $h_{i_{j'}}^1$ was added to $\overline{J}_{t'}$
		from  ${J}_{t'} = \{h_{t'}^1, h_{t'}^2\}$ at the stage $t'$. As
		$h_{t}^1 \succ_2 h_{t'}^1$, $t < t'$. In other words,
		$h_{t'}^1$ was added to the $\BestEURes$ after
		$h_{t}^1$. As $V_{1h_{t}^1} < V_{1h_{t'}^1} = V_{1h_{t'}^2}$,
		$h_{t'}^2$ is also added to $\overline{J}_{t'}$
		at the stage $t'$. As $h_{t'}^1$ is half-house allocated to agent $1$,
		$h_{t'}^2$  was removed from $\overline{J}_{t''}$ at some later stage $t''$.
		However, it can not be removed before $h_{t}^1$
		which has a smaller utility. This leads to a
		contradiction as $h_{t}^1 \in \bestalloc_1$
		and $h_{t'}^2 \notin \bestalloc_1$.
		\end{proof}

		The next lemma shows that $\bestalloc_1$ is point-wise at most as good as $\bestalloc_2$
		with respect agent 2 preferences.

		\begin{lemma}\label{l:ordering}
		$\bestalloc_2(k) \succ_2 \bestalloc_1(k)$ or $\bestalloc_1(k)$ and $\bestalloc_2(k)$
	are halves of the same house $k=1,\ldots,2m$.
		\end{lemma}
		\begin{proof} By induction on the number of rounds. The base case holds trivially as
		$\bestalloc_2(1)$ is in $\{h_{1}^1, h_{1}^2\}$
		and $h_{1}^1 \succ_2 \bestalloc_1(1) $ or $\bestalloc_1(1) = h_1^1$ and
	$\bestalloc_2(1) = h_1^2$.

		Assume that the statement holds for $i-1$ rounds. Consider the $i$th round.

		Suppose, by contradiction, $h:= \bestalloc_1(i) \succ_2 \bestalloc_2(i)=:h'$
	and $h$ and $h'$ are not halves of the same house.
		As  $h$ and $h'$ are allocated
		houses at the  $i$th  round then these houses are top preferences
		of agent $1$ and agent $2$, respectively, after $i-1$th round.
		As $h \succ_2 h'$, there exists a round $i' < i$ such that $h$
		is the top preference of agent $2$ at this round. Moreover,
		$h$ is available to agent $2$ at this round as agent $1$ only requests
		it at the $i$th round. Hence, $h$ will be allocated to
		agent $2$ at the $i'$th round. This contradicts the assumption
		that $h$ is allocated to agent $1$.
		\end{proof}

		The next result is the key result the section on computing the best EU response. We consider half-houses $G_{H_1} =  \{h_{i_1}^1,\ldots, h_{i_p}^1\}$ allocated to agent $1$ and $G_{H_2} =  \{h_{i_1}^2,\ldots, h_{i_p}^2\}$ allocated to agent $2$.
	$G_{H_1}$ and $G_{H_2}$ are lexicographically ordered.
	We show that, first, $h_{i_1}^1$ and $h_{i_1}^2$
		are allocated to agent $1$ and agent $2$, respectively, after that, $h_{i_2}^1$ and $h_{i_2}^2$ are allocated and so on.

		\begin{lemma}\label{l:strict_order_of_halves}
		Suppose houses in $G_{H_1}$  are allocated in rounds $k_{i_1}^1,\ldots, k_{i_p}^1$
		and  houses in $G_{H_2}$  are allocated in rounds $k_{i_1}^2,\ldots, k_{i_p}^2$.
		Then  $k_{i_1}^1 < k_{i_1}^2 <  k_{i_2}^1 < k_{i_2}^2 < \ldots < k_{i_p}^1 < k_{i_p}^2$.
		\end{lemma}
		\begin{proof}
		By contradiction, suppose that $h_{t}^1$ is the first half-house
		allocated to agent $1$ that violates the statement so that $k_{t}^1 < k_{t'}^1 < \ldots < k_{t}^2$.
		In other words, first, agent $1$ gets $h_t^1$ at the $k_{t}^1$th round
		and $h_{t'}^1$, which is a half of another house $h_{t'}$, at the $k_{t'}^1$th round, and later
		agent $2$ gets $h_t^2$ at the $k_{t}^2$th round.

		\begin{claim}\label{Claim 1.}
		The following inequality holds: $$h_{t}^1 \succ_2 h_{t'}^1.$$
		\end{claim}
	\begin{proof} This follows from the fact that houses in $\bestalloc_1 = \BestEURes$ are lexicographically ordered
		and the fact that $h_{t}^1$ is allocated before $h_{t'}^1$ to agent $1$.
		\end{proof}

	\begin{claim}
		\label{Claim 2.} The following inequality holds:
		$$ k_{t}^2 <  k_{t'}^{2}.$$
		\end{claim}
	\begin{proof}
	This follows from the structure $\prefcloned_2$ and $h_{t}^1 \succ_2 h_{t'}^1$ (Claim~\ref{Claim 1.}).
	\end{proof}
	From Claim~\ref{Claim 2.} and our assumption hypothesis we have
		$$k_{t}^1 < k_{t'}^1 < \ldots < k_{t}^2 <  k_{t'}^{2}.$$

		Suppose, $h_p^i$ and $h_q^i$ are allocated to agent $1$ at rounds
		$k_{t}^2$ and $k_{t'}^2$, respectively.

		\begin{claim}
		\label{Claim 3.} The following inequality holds:
		$$h_{t}^1 \succ_2  h_{t'}^1 \succ_2   h_p^i \succ_2 h_q^i.$$

		\end{claim}
		\begin{proof}
		Follows from Claim~\ref{Claim 2.}, $k_{t}^1 < k_{t'}^1 < \ldots < k_{t}^2 <  k_{t}^{2'}$,
		and the fact that houses in $\bestalloc_1 = \BestEURes$ are lexicographically ordered.
		\end{proof}

		We schematically show an allocation in the relevant rounds in the following table.
		The top part of the table shows allocation at rounds
		$k_{t}^1, k_{t'}^1, k_{t}^2 $  and $k_{t}^{2'}$.
		We use $\bullet$ to indicate that a house
		is allocated at a certain round but its label is not important for the proof.
	\begin{center}
		\scalebox{0.5}{
		\Large	
		$\begin{array}{|cccccccccccc|}
		\hline
		\multicolumn{12}{|c|}{\mathrm{Rounds}} \\

		&  \ldots & k_{t}^1& \ldots & k_{t'}^1  & \ldots & k_{t}^2 &  k_{t}^2+1& \ldots & k_{t'}^2-1 & k_{t'}^2  & \ldots \\
		\hline
	\multicolumn{12}{|c|}{\mathrm{An\ allocation\ obtained\ from\ } $\BestEUResAlgo$} \\

		\hline
		\bestalloc_1& \{\ldots  & h_{t}^1& \ldots & h_{t'}^1  & \ldots & h_{p}^i & \bullet&\ldots&\bullet& h_{q}^i & \ldots \} \\
		\bestalloc_2&  \{\ldots &  \bullet & \ldots & \bullet  & \ldots & h_{t}^2 & h_{s}^i&\ldots& h_{r}^i& h_{t'}^2 & \ldots \}\\
		\hline
		\multicolumn{12}{|c|}{\mathrm{ New\ allocation } } \\
		\hline
		 \bestalloc_1\cup \{h_{t}^2\} \setminus \{h_{t'}^1\}&\{\ldots  & h_{t}^1&\ldots  & h_{t}^2  & \ldots & h_{p}^1 & \bullet&\ldots&\bullet& h_{q}^1 & \ldots \} \\
		 \bestalloc_2\cup \{h_{t'}^1\} \setminus \{h_{t}^2\}& \{\ldots &   \bullet & \ldots &  \bullet &  \ldots&  h_{s}^i&\bullet& \ldots & h_{t'}^1& h_{t'}^2 & \ldots \}\\\hline

		\end{array}
		 $
		}
	\captionof{table}{A schematic representation of the proof of Claim~\ref{Claim 5.}}
	\end{center}

	\begin{claim}
		\label{Claim 4.} The following inequality holds:
		$$V_{1h_{t}^1} \geq V_{1h_{t'}^1}.$$
		\end{claim}
		\begin{proof}
		Follows from Claim~\ref{Claim 1.} and Proposition~\ref{prop:utility_order}.
	\end{proof}

		Next we show that agent $1$ can improve his outcome by deviating from $\bestalloc_1$ and obtain a contradiction to the assumption that $\bestalloc_1$  is a best response.

		\begin{claim}
		\label{Claim 5.}  If agent $1$ requests $h_t^2$ instead of $h_{t'}^1$ at the $k_{t'}^1$ round
		then agent $1$ improves its outcome.
	\end{claim}

		\begin{proof}
		First, we note that $h_t^2$ is available for agent $1$ at the $k_{t'}^1$ round.
		Indeed, by our assumption $k_{t'}^1 < k_{t}^2$, hence,
		the house $h_{t}^2$ is available to agent $1$ at round $k_{t'}^1$.
		Second, we show that  even if agent $1$ takes $h_t^2$ instead of $h_{t'}^1$ at the $k_{t'}^1$th round,
		agent $1$ can get all houses $\bestalloc_1 \setminus \{h_{t'}^1\}$. This shows that
		agent $1$ improves his outcome.

		From Claim~\ref{Claim 3.}, $h_{t'}^1 \succ_2 h_p^i$.
	From the structure of  $\prefcloned_2$ we know that $ \ldots \succ_2 h_{t'}^1 \succ_2 h_{t'}^2 \succ_2 \ldots$.
		Hence, due to Lemma~\ref{l:ordering}, during rounds  $k_{t}^2, \ldots,  k_{t'}^2-1$,
		the top houses of agent $2$ are ranked higher than $h_p^i$ in his profile.
		Also, the house $h_t^2$ is not available to agent $2$ at the $k_{t}^2$ round.
		Hence, agent $2$ is allocated the same houses in rounds  $k_{t}^2, \ldots,  k_{t'}^2-2$  as he was allocated before
		the change during rounds  $k_{t}^2+1, \ldots,  k_{t'}^2-1$ (see the second part of the table above).

		Consider the round $k_{t'}^2-1$.  As agent $1$ was not allocated $h_{t'}^1$
		at the $k_{t'}^1$th round, $h_{t'}^1$ is available for agent $2$
		at the $k_{t'}^2-1$-th round. Hence, agent $2$ is allocated $h_{t'}^1$
		at the $k_{t}^{2'}-1$-th round and $h_{t'}^2$ at the $k_{t'}^2$th round.
		The remaining rounds are identical to allocation using $\bestalloc_1$.
		The new allocation of agent $1$ is $\bestalloc_1  \cup \{h_{t}^2\} \setminus \{h_{t'}^1\}$
		which is strictly better than $\bestalloc_1$.
	\end{proof}
		Claim~\ref{Claim 5.} shows that  agent $1$ can improve his outcome and $\BestEURes$ is not a best response.
		This leads to a contradiction.\qed
		\end{proof}

		We denote  $\bestalloc_1^{-1}(h_i^{2})$ the round when $h_i^{2}$ is allocated.
		\begin{definition}\label{def:matching}
		A pair $\bestalloc_1$ and $\bestalloc_2$ has the \emph{matching} property if and only if for each pair of half-houses
		$h_i^1$ and $h_i^2$ such that $h_i^1 \in \bestalloc_1$ and $h_i^{2} \in \bestalloc_2$,
		 we have $\bestalloc_1^{-1}(h_i^1) = \bestalloc_2^{-1}(h_i^{2})$.
		\end{definition}

		\begin{example}\label{exm:matching}
		Consider
		$\bestalloc_1 = \{h_1^1,h_5^1, h_5^2,  h_3^1,  h_6^1, h_6^2\}$
		and
		$\bestalloc_2= \{h_1^2, h_2^1, h_2^2, h_3^2, h_4^1, h_4^2\}$.
		These profiles have the matching property as
		$\bestalloc_1^-1(h_1^1) = \bestalloc_2^{-1}(h_1^{2})$
		and
		$\bestalloc_1^-1(h_3^1) = \bestalloc_2^{-1}(h_3^{2})$.

		Consider
		$\bestalloc_1 = \{h_1^1, h_3^1, h_5^1, h_5^2, h_6^1, h_6^2\}$
		and
		$\bestalloc_2= \{h_1^2, h_2^1, h_2^2, h_3^2, h_4^1, h_4^2\}$.
		These profiles do not have the matching property as $\bestalloc_1^-1(h_3^1)\neq \bestalloc_2^{-1}(h_3^{2})$.
		\end{example}

		\begin{lemma}\label{l:re_ordering}
		For any $\bestalloc_1$ there exists $\bestalloc_1'$ that has the consecutivity property
		and such the pair $\bestalloc_1'$ and $\bestalloc_2$ has the matching property.
		Moreover, the allocation obtained by agent $1$ using $\bestalloc_1'$ is the same
		as the allocation obtained using $\bestalloc_1$.
		\end{lemma}
		\begin{proof}
		We set $\bestalloc_1' = \bestalloc_1$. Note that $\bestalloc_1'$
		has the consecutivity property as  $\bestalloc_1$ does as
	by Proposition~\ref{prop:consec} if a full-house of $h_i$ is allocated to $1(2)$ then $h_i^1$ and $h_i^2$ are allocated to $1(2)$
		in two consecutive rounds.

		Suppose, the pair $\bestalloc_1'$ and $\bestalloc_2$
		satisfies the statement up to round $k_{t}^1$.
		As $\bestalloc_1'$ has the consecutivity property,
		only the matching property can fail:
		$h_t^1$ is allocated to agent $1$ at the $k_{t}^1$ round and
		$h_t^2$  is allocated to agent $2$  at the $k_{t}^2$ round and
		$k_{t}^1 < k_{t}^2$.

		We show that we can move $h_t^1$ to round $k_{t}^2$
		and move all houses allocated during round $k_{t}^1+1, \ldots, k_{t}^2$
		one round forward in $\bestalloc_1'$. These shifts preserve the same allocation
		for agent $1$ and agent $2$ and the consecutivity property.

		By Lemma~\ref{l:strict_order_of_halves} we know that
		none of the half-houses are allocated to agent $1$
		during rounds $k_{t}^1+1,\ldots, k_{t}^2$.
		Hence, only full houses are allocated between these rounds.
		This means that the number of rounds between
		$k_{t}^1+1$ and $ k_{t}^2$ is even or 0.

		We also observe that none of the half houses
		are allocated to agent $2$ between rounds $k_{t}^1+1$ and $k_{t}^2$
		as $h_{t}^2$ is the first half-house allocated
		to agent $2$ after round $k_{t}^1$. Moreover,
		agent $2$ is not allocated houses greater  than $h_{t}^2$
		during rounds $k_{t}^1+1, \ldots, k_{t}^2$.

		We move the house $h_{t}^1$ to the position $k_{t}^2$
		and shift all houses in positions  $k_{t}^1+1$ and $ k_{t}^2$
		one round forward in $\bestalloc_1'$. Note that we preserve
		consecutivity property as all halves are moved together.

		After the move, agent $1$ still gets the same houses
		in rounds $k_{t}^1,\ldots,k_{t}^2$ as
		shifted houses are allocated even in earlier rounds
		compared to $\bestalloc_1$ and agent $1$ is allocated
		$h_{t}^1$ in the same round as agent $2$.
		Hence, allocations up to the round $k_{t}^2$
		are identical for $\bestalloc_1$ and $\bestalloc_1'$
		and both consecutivity and matching properties hold.

		We repeat the argument for the next half-house
		that violates the statement.
		\end{proof}

		\begin{example}\label{exm:swap}
		$\bestalloc_1 = \{h_1^1, h_3^1, h_5^1, h_5^2, h_6^1, h_6^2\}$
		and
		$\bestalloc_2= \{h_1^2, h_2^1, h_2^2, h_3^2, h_4^1, h_4^2\}$.
		We do not need to move $h_1^1$ as it is matched with $h_1^2$.
		We move $h_3^1$ to the fourth round so that  it is allocated
		at the same round as $h_3^2$.
	\begin{center}	
	\scriptsize
	$\begin{array}{|ccccccc|}
		\hline
		\multicolumn{7}{|c|}{\mathrm{Rounds}} \\

		 & 1 & 2& 3& 4 & 5 & 6 \\
		\hline
	\multicolumn{7}{|c|}{\mathrm{An\ allocation\ obtained\ from\ } $\BestEUResAlgo$} \\
		\hline
		\bestalloc_1 & \{ h_1^1,  & \mathbf{h_3^1 },& h_5^1, & h_5^2, & h_6^1, & h_6^2\}\\
		\bestalloc_2 & \{ h_1^2, & h_2^1, & h_2^2,& h_3^2,& h_4^1,& h_4^2\} \\
		\hline
		\multicolumn{7}{|c|}{\mathrm{New\ allocation\ with\ the \ matching\ property}} \\
		\hline
		\bestalloc_1' & \{ h_1^1,  & h_5^1, & h_5^2, & \mathbf{h_3^1}, &  h_6^1, & h_6^2\}\\
		\bestalloc_2 &  \{h_1^2, & h_2^1, & h_2^2,& h_3^2,& h_4^1,& h_4^2 \}\\
		\hline
		\end{array}
		 $
	\captionof{table}{A schematic representation of Example~\ref{exm:swap}.}
	\end{center}
		\end{example}

	A proof of Lemma~\ref{l:re_ordering} gives an correctness argument
	for lines~\ref{alg:delaybegin_loop}--\ref{alg:delayend_loop} in Algorithm~\ref{algo:2agent-EU-BR}.
	In these lines we put half-houses allocated to agent 1 later in the ordering
	to ensure that the matching property holds, i.e. agents obtain half-houses in the same rounds.

		\begin{lemma}\label{l:re_cons}
		Consider $\bestalloc_1$ and $\bestalloc_2$
		that satisfy  consecutivity and matching properties.
		Then there exists a preference  $\prefclonedbest_1$ over $\hclone$
		for agent $1$ that has the  consecutivity property and gives
		the same allocation as $\bestalloc_1$.
		\end{lemma}
		\begin{proof}
		Given $\bestalloc_1$ that satisfies properties in the statement
		of the lemma, we build a preference  $\prefclonedbest_1$ in the following way.
		We keep houses as they are ordered in  $\bestalloc_1$. For each half-house $h_i^1$
		allocated to agent $1$ we rank $h_i^2$ right after $h_i^1$. We put houses that are not allocated
		to agent $1$ in an arbitrary order, keeping halves together, at the end of the profile.
		Note that inserting $h_i^2$  after $h_i^1$ does not change the allocation
		as we know that $h_i^2$ is allocated to agent $2$ in the same round as $h_i^1$
		is allocated to agent $1$. Hence, $h_i^2$ will never be the top element
		for agent $1$ at any round. Hence,  $\prefclonedbest_1$ gives the same allocation as
		$\bestalloc_1$.
		\end{proof}
	A proof of Lemma~\ref{l:re_cons} provides a correctness argument for lines~\ref{alg:conbegin_loop}--
	\ref{alg:conend_loop}  in Algorithm~\ref{algo:2agent-EU-BR}. In these lines we move half-houses obtained by agent 2 right after corresponding half-houses obtained by agent 2.

		By Lemma~\ref{l:altbr2psbr}, given ${\prefclonedbest_1}$ which is the best response for $\prefcloned_2$,
		that satisfies the consecutivity property, $\prefbest_1$ obtained by the order-preserving join from  ${\prefclonedbest_1}$ is the best response for $\pref_2$
		using $\ps$.

	\begin{theorem}
		For the case of two agents and the PS rule, a DL best response and an EU best response are equivalent.
	\end{theorem}
	\begin{proof}
		For two agents, PS assigns probabilities from the set $\{0,1/2,1\}$. Hence DL preferences can be represented by the EU preferences where the utility are exponential: the utility of a more preferred house is twice the utility of the next preferred house. Hence a response if a DL best response if it is an EU best response for exponential utilities. On the other we have shown that for two agents and the PS rule, an EU best response is the same for any utilities compatible with the preferences. Hence for two agents, an EU best response for any utilities is the same as the EU best response for exponential utilities which in turn is the same as a DL best response.
	\end{proof}

	\section{Proof of Theorem~\ref{th:cycle}}


	\begin{proof}
	Using a computer program we have found the following 15 step sequence which leads
	to a cycling of the preference profile.  We use $U$ to denote
	the matrix of utilities of the agents over the items such that
	$U[1][1]$ is the utility of agent $1$ for house $h_1$.
	We use $P$ to represent the reported profile of each agent,
	$P[i][j]$ denotes the $j$th most preferred house of agent $i$.
	Note that $P$ starts as the truthful reporting in our example.
	We use $PS[i][j]$ to represent the fraction of house $j$ 
	that is eaten by agent $i$.  We use $EU[i]$ to be the expected
	utility of agent $i$.

	\noindent
	The initial preferences and utilities of the agents are

	\begin{minipage}{.4\textwidth}
	\[
	P_0=\begin{pmatrix}
		h_2 & h_3 & h_1 & h_4 & h_6 & h_5 \\
		h_6 & h_5 & h_2 & h_1 & h_4 & h_3 \\
		h_3 & h_6 & h_2 & h_1 & h_5 & h_4
	 \end{pmatrix}
	\]
	\end{minipage}
	\begin{minipage}{.48\textwidth}
	\[
	 U_0 = \begin{pmatrix}
		3 & 5 & 4 & 2 & 0 & 1 \\
		2 & 3 & 0 & 1 & 4 & 5 \\
		2 & 3 & 5 & 0 & 1 & 4
	 \end{pmatrix}.
	\]
	\end{minipage}
	\bigskip

	\noindent
	This yields the following allocation and utilities at the start

	\noindent
	\begin{minipage}{.30\textwidth}
	\[
	P_0=\begin{pmatrix}
		h_2 & h_3 & h_1 & h_4 & h_6 & h_5 \\
		h_6 & h_5 & h_2 & h_1 & h_4 & h_3 \\
		h_3 & h_6 & h_2 & h_1 & h_5 & h_4
	 \end{pmatrix}
	\]
	\end{minipage}
	\begin{minipage}{.48\textwidth}
	\[
	 PS_0 = \begin{pmatrix}
		1/2 	& 1	 & 0	 & 1/2	& 0		& 0 \\
		0 	& 0	 & 0	 & 1/4	& 3/4		& 1 \\
		1/2 	& 0	 & 1	 & 1/4	& 1/4		& 0
	 \end{pmatrix}
	\]
	\end{minipage}
	\begin{minipage}{.2\textwidth}
	\[
	 EU_0 = \begin{pmatrix}
		7.5  \\
	 	8.25  \\
	  	6.25
	 \end{pmatrix}.
	\]
	\end{minipage}
	\bigskip

	\noindent
	In Step 1, agent 3 changes his report and improves his utility.

	\noindent
	\begin{minipage}{.30\textwidth}
	\[
	P_1=\begin{pmatrix}
		 h_2 & h_3 & h_1 & h_4 & h_6 & h_5 \\
		 h_6 & h_5 & h_2 & h_1 & h_4 & h_3 \\
		 h_6 & h_3 & h_1 & h_2 & h_4 & h_5
	 \end{pmatrix}
	\]
	\end{minipage}
	\begin{minipage}{.48\textwidth}
	\[
	 PS_1 = \begin{pmatrix}
		5/12&1&1/4&1/3&0& 0 \\
	   1/6&0& 0&1/3&1&1/2 \\
	   5/12&0&3/4&1/3&0&1/2
	 \end{pmatrix}
	\]
	\end{minipage}
	\begin{minipage}{.2\textwidth}
	\[
	 EU_1 = \begin{pmatrix}
		7.9167  \\
	 	71667  \\
	  	6.5833
	 \end{pmatrix}.
	\]
	\end{minipage}
	\bigskip

	\noindent
	In Step 2, agent 1 changes his report in response.

	\noindent
	\begin{minipage}{.30\textwidth}
	\[
	P_2=\begin{pmatrix}
		h_3 & h_2 & h_1 & h_4 & h_5 & h_6 \\
		 h_6 & h_5 & h_2 & h_1 & h_4 & h_3 \\
		 h_6 & h_3 & h_1 & h_2 & h_4 & h_5
	 \end{pmatrix}
	\]
	\end{minipage}
	\begin{minipage}{.48\textwidth}
	\[
	 PS_2 = \begin{pmatrix}
		1/24&7/8&3/4&1/3&0& 0\\
	   1/24&1/8&0&1/3&1&1/2\\
	  11/12&0&1/4&1/3&0&1/2
	 \end{pmatrix}
	\]
	\end{minipage}
	\begin{minipage}{.2\textwidth}
	\[
	 EU_2 = \begin{pmatrix}
		8.1667  \\
	 	7.2917  \\
	  	5.0833
	 \end{pmatrix}.
	\]
	\end{minipage}
	\bigskip

	\noindent
	In Step 3, agent 3 again changes his report.

	\noindent
	\begin{minipage}{.3\textwidth}
	\[
	P_3=\begin{pmatrix}
		h_3 & h_2 & h_1 & h_4 & h_5 & h_6 \\
		 h_6 & h_5 & h_2 & h_1 & h_4 & h_3 \\
		 h_3 & h_6 & h_2 & h_1 & h_5 & h_4
	 \end{pmatrix}
	\]
	\end{minipage}
	\begin{minipage}{.48\textwidth}
	\[
	 PS_3 = \begin{pmatrix}
		1/2&5/8&1/2&3/8&0& 0\\
	    0& 0& 0&5/16&15/16&3/4\\
	   1/2&3/8&1/2&5/16&1/16&1/4
	 \end{pmatrix}
	\]
	\end{minipage}
	\begin{minipage}{.2\textwidth}
	\[
	 EU_3 = \begin{pmatrix}
		7.3750  \\
	 	7.8125  \\
	  	5.6875
	 \end{pmatrix}.
	\]
	\end{minipage}
	\bigskip

	\noindent
	In Step 4, agent 1 reacts again.

	\noindent
	\begin{minipage}{.3\textwidth}
	\[
	P_4=\begin{pmatrix}
		 h_2 & h_1 & h_3 & h_4 & h_5 & h_6 \\
		 h_6 & h_5 & h_2 & h_1 & h_4 & h_3 \\
		 h_3 & h_6 & h_2 & h_1 & h_5 & h_4
	 \end{pmatrix}
	\]
	\end{minipage}
	\begin{minipage}{.48\textwidth}
	\[
	 PS_4 = \begin{pmatrix}
		1/2&1& 0&1/2&0& 0\\
	    0& 0& 0&1/4&3/4&1\\
	   1/2&0& 1&1/4&1/4&0
	 \end{pmatrix}
	\]
	\end{minipage}
	\begin{minipage}{.2\textwidth}
	\[
	 EU_4 = \begin{pmatrix}
		7.500  \\
	 	8.250  \\
	  	6.250
	 \end{pmatrix}.
	\]
	\end{minipage}
	\bigskip

	\noindent
	In Step 5, agent 3 reacts again.

	\noindent
	\begin{minipage}{.3\textwidth}
	\[
	P_5=\begin{pmatrix}
		h_2 & h_1 & h_3 & h_4 & h_5 & h_6 \\
		 h_6 & h_5 & h_2 & h_1 & h_4 & h_3 \\
		 h_6 & h_2 & h_3 & h_1 & h_4 & h_5
	 \end{pmatrix}
	\]
	\end{minipage}
	\begin{minipage}{.48\textwidth}
	\[
	 PS_5 = \begin{pmatrix}
		7/8&3/4&1/16&5/16&0& 0 \\
	   1/8&0& 0&3/8&1&1/2 \\
	    0&1/4&15/16&5/16&0&1/2
	 \end{pmatrix}
	\]
	\end{minipage}
	\begin{minipage}{.2\textwidth}
	\[
	 EU_5 = \begin{pmatrix}
		7.250  \\
	 	7.125  \\
	  	7.4375
	 \end{pmatrix}.
	\]
	\end{minipage}
	\bigskip

	\noindent
	In Step 6, agent 2 reacts.

	\noindent
	\begin{minipage}{.3\textwidth}
	\[
	P_6=\begin{pmatrix}
		h_2 & h_1 & h_3 & h_4 & h_5 & h_6 \\
		 h_6 & h_2 & h_1 & h_5 & h_3 & h_4 \\
		 h_6 & h_2 & h_3 & h_1 & h_4 & h_5
	 \end{pmatrix}
	\]
	\end{minipage}
	\begin{minipage}{.48\textwidth}
	\[
	 PS_6 = \begin{pmatrix}
		1/2&2/3&1/4&1/2&1/12&0 \\
	   1/2&1/6&0& 0&5/6&1/2 \\
	    0&1/6&3/4&1/2&1/12&1/2
	 \end{pmatrix}
	\]
	\end{minipage}
	\begin{minipage}{.2\textwidth}
	\[
	 EU_6 = \begin{pmatrix}
		6.833  \\
	 	7.333  \\
	  	6.333
	 \end{pmatrix}.
	\]
	\end{minipage}
	\bigskip

	\noindent
	In Step 7, agent 3 reacts.

	\noindent
	\begin{minipage}{.3\textwidth}
	\[
	P_7=\begin{pmatrix}
		h_2 & h_1 & h_3 & h_4 & h_5 & h_6 \\
		 h_6 & h_2 & h_1 & h_5 & h_3 & h_4 \\
		 h_6 & h_3 & h_1 & h_2 & h_5 & h_4
	 \end{pmatrix}
	\]
	\end{minipage}
	\begin{minipage}{.48\textwidth}
	\[
	 PS_7 = \begin{pmatrix}
		1/2&3/4&1/8&5/8&0& 0\\
	   1/2&1/4&0&3/16&9/16&1/2\\
	    0& 0&7/8&3/16&7/16&1/2
	 \end{pmatrix}
	\]
	\end{minipage}
	\begin{minipage}{.2\textwidth}
	\[
	 EU_7 = \begin{pmatrix}
		7.00  \\
	 	6.6875  \\
	  	6.8125
	 \end{pmatrix}.
	\]
	\end{minipage}
	\bigskip

	\noindent
	In Step 8, agent 1 changes his report.

	\noindent
	\begin{minipage}{.3\textwidth}
	\[
	P_8=\begin{pmatrix}
		h_2 & h_3 & h_1 & h_4 & h_5 & h_6 \\ 
	 	h_6 & h_2 & h_1 & h_5 & h_3 & h_4 \\
	 	h_6 & h_3 & h_1 & h_2 & h_5 & h_4
	 \end{pmatrix}
	\]
	\end{minipage}
	\begin{minipage}{.48\textwidth}
	\[
	 PS_8 = \begin{pmatrix}
		5/24&3/4&3/8&2/3&0& 0\\
	   7/12&1/4&0&1/6&1/2&1/2\\
	   5/24&0&5/8&1/6&1/2&1/2
	 \end{pmatrix}
	\]
	\end{minipage}
	\begin{minipage}{.2\textwidth}
	\[
	 EU_8 = \begin{pmatrix}
		7.2083  \\
	 	6.5833  \\
	  	6.0417
	 \end{pmatrix}.
	\]
	\end{minipage}
	\bigskip

	\noindent
	In Step 9, agent 2 reacts.

	\noindent
	\begin{minipage}{.3\textwidth}
	\[
	P_9=\begin{pmatrix}
		h_2 & h_3 & h_1 & h_4 & h_5 & h_6 \\
	 	h_6 & h_2 & h_5 & h_1 & h_3 & h_4 \\
	 	h_6 & h_3 & h_1 & h_2 & h_5 & h_4
	 \end{pmatrix}
	\]
	\end{minipage}
	\begin{minipage}{.48\textwidth}
	\[
	 PS_9 = \begin{pmatrix}
		1/2&3/4&3/8&3/8&0& 0 \\
	    0&1/4&0&5/16&15/16&1/2 \\
	   1/2&0&5/8&5/16&1/16&1/2
	 \end{pmatrix}
	\]
	\end{minipage}
	\begin{minipage}{.2\textwidth}
	\[
	 EU_9 = \begin{pmatrix}
		7.5  \\
	 	7.3125  \\
	  	6.1875
	 \end{pmatrix}.
	\]
	\end{minipage}
	\bigskip

	\noindent
	In Step 10, agent 3 reacts again.

	\noindent
	\begin{minipage}{.3\textwidth}
	\[
	P_{10}=\begin{pmatrix}
		h_2 & h_3 & h_1 & h_4 & h_5 & h_6 \\
	 	h_6 & h_2 & h_5 & h_1 & h_3 & h_4 \\
	 	h_3 & h_1 & h_2 & h_5 & h_4 & h_6
	 \end{pmatrix}
	\]
	\end{minipage}
	\begin{minipage}{.48\textwidth}
	\[
	 PS_{10} = \begin{pmatrix}
		1/2&1& 0&1/2&0& 0 \\
	    0& 0& 0&1/4&3/4&1 \\
	   1/2&0& 1&1/4&1/4&0
	 \end{pmatrix}
	\]
	\end{minipage}
	\begin{minipage}{.2\textwidth}
	\[
	 EU_{10} = \begin{pmatrix}
		7.5  \\
	 	8.25  \\
	  	6.25
	 \end{pmatrix}.
	\]
	\end{minipage}
	\bigskip

	\noindent
	In Step 11, agent 2 reacts.

	\noindent
	\begin{minipage}{.3\textwidth}
	\[
	P_{11}=\begin{pmatrix}
		h_2 & h_3 & h_1 & h_4 & h_5 & h_6 \\
	 	h_5 & h_2 & h_3 & h_6 & h_1 & h_4 \\
	 	h_3 & h_1 & h_2 & h_5 & h_4 & h_6
	 \end{pmatrix}
	\]
	\end{minipage}
	\begin{minipage}{.48\textwidth}
	\[
	 PS_{11} = \begin{pmatrix}
		1/2&1& 0&1/2&0& 0 \\
	    0& 0& 0& 0& 1& 1 \\
	   1/2&0& 1&1/2&0& 0
	 \end{pmatrix}
	\]
	\end{minipage}
	\begin{minipage}{.2\textwidth}
	\[
	 EU_{11} = \begin{pmatrix}
		7.5  \\
	 	9.0  \\
	  	6.0
	 \end{pmatrix}.
	\]
	\end{minipage}
	\bigskip

	\noindent
	In Step 12, agent 3 reacts.

	\noindent
	\begin{minipage}{.3\textwidth}
	\[
	P_{12}=\begin{pmatrix}
		h_2 & h_3 & h_1 & h_4 & h_5 & h_6 \\
	 	h_5 & h_2 & h_3 & h_6 & h_1 & h_4 \\
	 	h_3 & h_2 & h_5 & h_6 & h_1 & h_4
	 \end{pmatrix}
	\]
	\end{minipage}
	\begin{minipage}{.48\textwidth}
	\[
	 PS_{12} = \begin{pmatrix}
		2/3&1& 0&1/3&0& 0 \\
	   1/6&0& 0&1/3&1&1/2 \\
	   1/6&0& 1&1/3&0&1/2
	 \end{pmatrix}
	\]
	\end{minipage}
	\begin{minipage}{.2\textwidth}
	\[
	 EU_{12} = \begin{pmatrix}
		7.6667  \\
	 	7.1667  \\
	  	7.3333
	 \end{pmatrix}.
	\]
	\end{minipage}
	\bigskip

	\noindent
	In Step 13, agent 2 reacts.

	\noindent
	\begin{minipage}{.3\textwidth}
	\[
	P_{13}=\begin{pmatrix}
		h_2 & h_3 & h_1 & h_4 & h_5 & h_6 \\
	 	h_6 & h_2 & h_3 & h_5 & h_1 & h_4 \\
	 	h_3 & h_2 & h_5 & h_6 & h_1 & h_4
	 \end{pmatrix}
	\]
	\end{minipage}
	\begin{minipage}{.48\textwidth}
	\[
	 PS_{13} = \begin{pmatrix}
		2/3&1& 0&1/3&0& 0 \\
	   1/6&0& 0&1/3&1/2&1 \\
	   1/6&0& 1&1/3&1/2&0
	 \end{pmatrix}
	\]
	\end{minipage}
	\begin{minipage}{.2\textwidth}
	\[
	 EU_{13} = \begin{pmatrix}
		7.6667  \\
	 	7.6667  \\
	  	5.8333
	 \end{pmatrix}.
	\]
	\end{minipage}
	\bigskip

	\noindent
	In Step 14, agent 3 reacts again.

	\noindent
	\begin{minipage}{.3\textwidth}
	\[
	P_{14}=\begin{pmatrix}
		h_2 & h_3 & h_1 & h_4 & h_5 & h_6 \\
	 	h_6 & h_2 & h_3 & h_5 & h_1 & h_4 \\
	 	h_3 & h_1 & h_2 & h_5 & h_4 & h_6
	 \end{pmatrix}
	\]
	\end{minipage}
	\begin{minipage}{.48\textwidth}
	\[
	 PS_{14} = \begin{pmatrix}
		1/2&1& 0&1/2&0& 0 \\
	    0& 0& 0&1/4&3/4&1 \\
	   1/2&0& 1&1/4&1/4&0
	 \end{pmatrix}
	\]
	\end{minipage}
	\begin{minipage}{.2\textwidth}
	\[
	 EU_{14} = \begin{pmatrix}
		7.5  \\
	 	8.25  \\
	  	6.25
	 \end{pmatrix}.
	\]
	\end{minipage}
	\bigskip

	\noindent
	In Step 15, agent 2 reacts once more to agent 3.

	\noindent
	\begin{minipage}{.3\textwidth}
	\[
	P_{15}=\begin{pmatrix}
		h_2 & h_3 & h_1 & h_4 & h_5 & h_6 \\
	 	h_5 & h_2 & h_3 & h_6 & h_1 & h_4 \\
	 	h_3 & h_1 & h_2 & h_5 & h_4 & h_6
	 \end{pmatrix}
	\]
	\end{minipage}
	\begin{minipage}{.48\textwidth}
	\[
	 PS_{15} = \begin{pmatrix}
		1/2&1& 0&1/2&0& 0 \\
	    0& 0& 0& 0& 1& 1 \\
	   1/2&0& 1&1/2&0& 0
	 \end{pmatrix}
	\]
	\end{minipage}
	\begin{minipage}{.2\textwidth}
	\[
	 EU_{15} = \begin{pmatrix}
		7.5  \\
	 	9.0  \\
	  	6.0
	 \end{pmatrix}.
	\]
	\end{minipage}
	\bigskip

	\noindent
	This last step is the same profile as step 11, which means we have 
	cycled.
	\end{proof}

	\end{document}